\documentclass[11pt,leqno]{amsart}

\usepackage{bbm}
\usepackage{color}
\usepackage{graphicx}
\usepackage{amsmath,amsthm,amssymb,verbatim,amsfonts}
\usepackage{tensor}
\usepackage{slantsc}
\usepackage{slashed}

\usepackage{mathrsfs}
\usepackage{mathtools}
\usepackage{abstract}

\usepackage{bm}

\usepackage[sort,nocompress]{cite}

\usepackage{enumerate}


\usepackage{etoolbox}
\makeatletter
\patchcmd{\@maketitle}{\newpage}{}{}{} 
\makeatother

\usepackage{mathtools}

\usepackage{marginnote}

\newcommand{\red}[1]{#1}

\DeclareFontFamily{OT1}{rsfs}{}
\DeclareFontShape{OT1}{rsfs}{m}{n}{ <-7> rsfs5 <7-10> rsfs7 <10-> rsfs10}{}
\DeclareMathAlphabet{\mycal}{OT1}{rsfs}{m}{n}

\newcommand{\mcB}{{\mycal B}}
\newcommand{\mcL}{{\mycal L}}

\newcommand{\mcE}{{\mycal E}}
\newcommand{\mcG}{{\mycal G}}

\newcommand{\mcP}{{\mycal P}}

{\catcode `\@=11 \global\let\AddToReset=\@addtoreset}
\AddToReset{equation}{section}

\newcounter{mnotecount}[section]

\newcommand{\R}{\mathbb R}
\newcommand{\N}{\mathbb N}

\newcommand{\supp}{\operatorname{supp}}

\theoremstyle{plain}
\newtheorem{thm}{Theorem}[section]
\newtheorem{lem}[thm]{Lemma}
\newtheorem{prop}[thm]{Proposition}
\newtheorem{cor}[thm]{Corollary}

\theoremstyle{definition}
\newtheorem{defn}[thm]{Definition}

\newtheorem{rem}[thm]{Remark}

\numberwithin{equation}{section}

\newcommand{\nn}{\nonumber}

\newcommand{\tr}{{\rm tr}}
\newcommand{\trg}{\tr_g}

\newcommand{\blue}[1]{{{#1}}}

\newcommand{\vertiii}[1]{{\left\vert\kern-0.25ex\left\vert\kern-0.25ex\left\vert #1 
    \right\vert\kern-0.25ex\right\vert\kern-0.25ex\right\vert}}

\newcommand{\const}{\mbox{\rm const}} %constants

\newcommand{\mnote}[1]%{}
{\protect{\stepcounter{mnotecount}}$^{\mbox{\footnotesize
$%\!\!\!\!\!\!\,
\bullet$\themnotecount}}$ \marginpar{%\color{red}%
\raggedright\tiny
$\!\!\!\!\!\!\,\bullet$\themnotecount: #1} }

\newcommand{\absol}[1]{|#1|}
\newcommand{\absolg}[1]{|#1|_g}
\newcommand{\absolgal}[1]{\left|#1\right|_g}
\newcommand{\absolgdot}[1]{\left|#1\right|_{\dot{g}}}

\newcommand{\less}{\lesssim}

\newcommand{\divg}{{\rm div}_g}
\newcommand{\dvol}{d\mu_g}

\newcommand{ \scalpr}[2]{{\langle #1, #2 \rangle}_g}

\newcommand{\hb}{\mathbf{h}}
\newcommand{\bcov}{\widetilde{\nabla}}

\newcommand{\volg}{d\mu_g}

\newcommand{\phat}{\widehat{p}}
\newcommand{\pbar}{\underline{p}}

\newcommand{\nullhat}{\hat{0}}

\newcommand{\xhat}{\widehat{X}}
\newcommand{\nhat}{\widehat{N}}

\newcommand{\F}{\mathfrak{F}}

\newcommand{\riem}{\text{Riem}[g]}

\newcommand{\pnull}{\p_{e_0}}

\newcommand{\urescov}{\overline{D}}

\newcommand{\fixcov}{ \widehat{D} }

\newcommand{\difften}[3]{\Upsilon^#1_{#2 #3} } 

\newcommand{\p}{\partial}

\newcommand{\jnull}{\mathbf{J}}
\newcommand{\jvec}{\mathcal{J}}

\newcommand{\Lnorm}[2]{\lVert#1\rVert_{L^{#2}}}

\newcommand{\Hnorm}[2]{\lVert#1\rVert_{H^{#2}}}
\newcommand{\Hnormadj}[2]{\left\lVert#1\right\rVert_{H^{#2}}}
\newcommand{\Hgnorm}[2]{\lVert#1\rVert_{H^{#2}(M[g])}}

\newcommand{\Fnorm}[2]{\vertiii{#1}_{H^{#2}}}

\newcommand{\sasag}{\underline{\mathbf{g}}}
\newcommand{\absolsasag}[1]{|#1|_{\sasag}}
\newcommand{\sasacov}{\mathbf{D}}

\newcommand{\A}{\mathbf{A}}
\newcommand{\B}{\mathbf{B}}

\newcommand{\energy}[2]{\mcE_{#1,#2}(f)}
\newcommand{\energysq}[2]{\mcE^2_{#1,#2}(f)}

\newcommand{\maxT}{{}^{\mathrm{M}}{} \mathbf{T}}
\newcommand{\vlasT}{^{\mathrm{V}}{\!}\mathbf{T} }

\newcommand{\bmrho}{\bm{\varrho}}
\newcommand{\EMax}{\mathbb{E}}
\newcommand{\Etot}{\mathbf{E}_{\mathrm{tot}}}

\newcommand{\epsilontot}{\varepsilon_{\mathrm{tot}}}

\begin{document}

\title[Stable cosmologies with collisionless charged matter]{Stable cosmologies with collisionless charged matter}
\author[\textsc{H.~Barzegar}, \textsc{D.~Fajman}]{\textsc{Hamed Barzegar}, \textsc{David Fajman}}
\address{
Gravitational Physics\\
Faculty of Physics\\
University of Vienna\\
Boltzmanngasse 5, 1090 Vienna \\
Austria}
\email{Hamed.Barzegar@univie.ac.at, David.Fajman@univie.ac.at}

\date{\today}

\maketitle
%%%%%%%%%%%%%%%%%%%%%%%%%%%%%%%%%%%%%%%%%%%%%%%%%%%%%%%%%%%%%%%%%%%%%%%
%%%%%%%%%%%%%%%%%%%%%%%%%%%%%%%%%%%%%%%%%%%%%%%%%%%%%%%%%%%%%%%%%%%%%%%
\begin{abstract}
It is shown that Milne models (a subclass of Friedmann--Lemaître--Robertson--Walker (FLRW) spacetimes with negative spatial curvature) are nonlinearly stable
in the set of solutions to the Einstein-Vlasov-Maxwell system, describing universes with ensembles of collisionless self-gravitating, charged particles. The system contains various slowly decaying borderline terms in the mutually coupled equations describing the propagation of particles and Maxwell fields. The effects of those terms are controlled using a suitable hierarchy based on the energy density of the matter fields.   
\end{abstract}
%%%

%\tableofcontents

\section{Introduction}
Studying the global dynamics of non-vacuum solutions to the Einstein equations is a major effort in general relativity with the aim to draw conclusions on isolated self-gravitating systems and cosmology.
In recent years several results have contributed substantially to the understanding of non-vacuum dynamics with self-gravitating matter models, which provide realistic features of matter in the actual universe such as relativistic fluids of various types and kinetic matter models. While fluids are known to require expansion of spacetime to avoid shock formation (cf.~\cite{RoSp13}) Vlasov matter shows a more regular behavior. When coupled to the Einstein equations, it is expected to exhibit only those types of singularity formation which are caused by gravity. Indeed, most nonlinear stability results for the vacuum Einstein equations have recently been generalized to the Einstein--Vlasov system, most prominently for the de~Sitter type spacetimes \cite{Ri13}, Minkowski spacetime \cite{BFJST20,FJS17,LT20,Ta17} and the Milne model \cite{AF17} and lower dimensional analogues \cite{F16-2,F17-2,F17,F18}.

The Einstein-Vlasov system describes spacetimes containing ensembles of self-gravitating collisionless particles and provides a realistic description of the large scale structure of spacetime. It admits various steady states modelling isolated self-gravitating matter configurations such as galaxies and galaxy clusters \cite{ReRe93,AFT15,AFT17} and similarly a variety of matter dominated cosmological models \cite{BFH20,Nu10}. In contrast to those models, in the stability analysis of \cite{Ri13,BFJST20,FJS17,LT20,Ta17,AF17} the matter distribution is considered to be small and disperses in the course of the evolution while the spacetime geometry asymptotes to the background vacuum geometry.

While the Einstein--Vlasov system models purely gravitative effects there exist generalizations, which include more detailed physical phenomena such as charged collisionless particles modeled by the Einstein--Vlasov--Maxwell system on which we focus in the following. 

\subsection{The Einstein--Vlasov--Maxwell system}

The Einstein--Vlasov--Maxwell system (EVMS)
\begin{equation}\label{EVM}
\begin{aligned}
\mathrm{Ric}[\mathbf{h}]-\frac12\mathrm{R}[\mathbf h]\cdot \mathbf h
 &=
  8 \pi \mathbf T
  \,,
 \\
  L_{\mathbf{h}, F}   f 
 &=0
 \,,
\\
\mathbf{d} F
&=
0
\,,
\\
\mathbf{d} \star F
&=
\star { J}
\,,
\end{aligned}
\end{equation}

describes spacetimes containing Maxwell fields and ensembles of collisionless charged particles, which interact via gravity and electromagnetism \cite{Dodelson2003,BinneyTremaine2008}.
Here, $\mathbf{h}$ denotes a Lorentzian metric on a given 4-manifold $\mathcal{M}$; $\mathbf{T} = {\vlasT} +{\maxT}$
is the total energy-momentum tensor with
 ${\vlasT}$ and $ {\maxT}$ denoting the energy-momentum tensors of the Vlasov matter and the Maxwell field, respectively; $L_{\mathbf{h}, F}$ is the  Liouville--Vlasov operator which consists of the geodesic spray and the Maxwell term; $\mathbf{d}$ and $\star$  denote the exterior derivative and the Hodge star operator on $\mathcal{M}$, respectively; $F$ is the Faraday tensor, and ${J}$ is the matter current $1$-form.

When setting the Faraday tensor $F$ and the charge to zero, \eqref{EVM} reduces to the Einstein--Vlasov system while setting the distribution function $f$ to zero it reduces to the Einstein--Maxwell system.

There exist only few results on the EVMS concerning stationary solutions \cite{AER09,T19,T20} and evolution in spherical symmetry \cite{N05}.
In particular, stability results have not been established yet in the class of solutions to the EVMS.

\subsection{Background spacetimes}

We consider a class of cosmological vacuum spacetimes which are Lorentz cones over a negative closed Einstein space $(M,\gamma)$ of dimension $3$ with Einstein constant $\kappa = -\tfrac{2}{9}$, i.e., $\text{Ric}[\gamma] = - \frac{2}{9} \gamma\,$,
where the value of $\kappa$ is chosen for convenience. Then, the \textit{Milne model}
\begin{equation}\label{eq: Milne geometry}
\big( (0, \infty) \times M, -dt^2 + \tfrac{t^2}{9} \cdot \gamma \big)
\end{equation}
is a solution to the vacuum Einstein equations. In contrast to the class of exponentially expanding spacetimes or those with power law inflation, the Milne model does not exhibit accelerated expansion. Its linear scale factor constitutes the threshold between accelerated and deccelerated expansion, which makes the model particularly interesting from the perspective of the regularizing effect of its expansion on various matter models \cite{FOW20}. Future stability of \eqref{eq: Milne geometry} under the Einstein flow has been established in the vacuum case in \cite{AM11} and for different types of matter in \cite{AF17,BFK19,FW19}. 

\subsection{Main result}
It is the purpose of the present paper to establish the first nonlinear stability result for the EVMS. In particular, we prove future nonlinear stability for the Milne model in the set of solutions to the EVMS. The main challenge in contrast to earlier stability results on the Milne model consists in the direct mutual coupling of the matter fields, which does induce various slowly decaying terms in the system. We discuss these aspects further below in detail. The main theorem reads as follows.

\begin{thm}\label{Theorem}
Let $(M, \gamma )$ be a compact, 3-dimensional negative Einstein manifold without boundary and Einstein constant $\kappa = - \tfrac{2}{9}$, and let $\bm{\epsilon} > 0$. Then, there exists a $\delta > 0$ such that for the rescaled initial data $ ({g_0}, {\Sigma_0},{f}_0, {\omega_0},{\dot{\omega}_0}) \in H^6 (M) \times H^5(M) \times H_{\rm{Vl}, 5,4,c}(TM) \times H^6 (M) \times H^5(M)$ at $t = t_0$ with compact momentum support of the initial particle distribution and
\begin{equation}
 ({g_0}, {\Sigma_0},{f}_0, {\omega_0},{\dot{\omega}_0}) \in   \mcB^{6,5,5,6,5}_{\delta} (\gamma, 0, 0, 0, 0)
 \,,
\end{equation}
the corresponding solution to the rescaled  Einstein--Vlasov--Maxwell system is future-global in time and future complete. Moreover, the rescaled metric and trace-free part of the second fundamental form  converge as 
\begin{equation}
 (g, \Sigma) \rightarrow (\gamma, 0) \quad \text{for} \quad  \tau \nearrow 0
 \,,
\end{equation}
with decay rates determined by $\bm{\epsilon}$ as in \eqref{eq: decay rates} below. In particular, the Milne model is future asymptotically stable for Einstein--Vlasov--Maxwell system in the class of initial data given above. 
\end{thm}

We will provide precise definitions of the objects in Theorem~\ref{Theorem} further below. However, to give an overview on the theorem, we briefly summarize them. We use $(g, \Sigma, f, \omega, \dot{\omega})$ to denote the rescaled variables: the Riemannian metric,  the trace-free part of the second fundamental, the distribution function, the spatial vector potential, and its time-derivative, respectively. $\tau < 0 $ denotes the mean curvature and is related to the time variable in \eqref{eq: Milne geometry} by $t = -3 \tau^{-1}$. Therefore, $\tau \nearrow 0$ corresponds to $t \rightarrow \infty$. $H^k (M)$ with $k \geq 0$ denotes the $L^2$-based Sobolev norm and $H_{\rm{Vl}, 5,4,c}(TM)$ denotes the space of distribution functions of compact momentum support on $TM$ corresponding to the standard $L^2$-Sobolev norms (cf.~\cite{F17}). We denote by $\mcB^{6,5,5,6,5}_{\delta} (\cdot, \cdot, \cdot, \cdot, \cdot)$ a ball of radius $\delta$ centered at its arguments in the set of $H^6 (M) \times H^5(M) \times H_{\rm{Vl}, 5,4,c}(TM) \times H^6 (M) \times H^5(M)$. 

\begin{rem}
The main theorem is related to foregoing works and more general settings in the following sense.
\begin{enumerate}
\item 
Theorem~\ref{Theorem} in particular implies Theorem~1 of \cite{AF17} (the Einstein-Vlasov case) and Theorem~7.1 of \cite{BFK19} (the Einstein-Maxwell case).
\item We consider the EVMS with particles of identical rest mass $m>0$ and identical charge $q$. The present result however directly generalizes to a collection of ensembles with different masses and charges (also of opposite signs). For simplicity
of the presentation we restrict ourselves to this specific case.
\end{enumerate}
\end{rem}
While there exists a number of stability results  for relativistic Vlasov-Maxwell systems \cite{B18,B19,B20} the present result is to our knowledge the first stability result for the EVMS.
\subsection{Aspects of the proof}
The proof of Theorem~\ref{Theorem} is partially based on the techniques developed to control the distribution function in \cite{F16,F17} and \cite{AF17} and the techniques to control the Maxwell fields derived in \cite{BFK19}. However, the present problem poses various new difficulties as Maxwell fields and distribution function are directly coupled via the Maxwell term in the transport equation and the particle-current term in the Maxwell equation. In fact, both terms are principal terms in the sense that they decay a priori at the slowest rate when compared with other perturbative terms in the respective equation. Moreover, the decay rates of these terms induce a loss in decay for the corresponding energies of the Maxwell fields and distribution function, respectively. We refer to those terms in the following as \emph{borderline terms}. 

In the Einstein equations a borderline term enters in the lapse equation \eqref{lapse} via the rescaled pressure $\eta$. This has already been observed and resolved in \cite{AF17} in conjunction with the corresponding borderline term $(*)$ in the transport equation \eqref{eq: Vlasov equation}. Their mutual coupling was controlled in \cite{AF17} by using the continuity equation to obtain a sharp estimate for the energy density of the distribution function.

In the presence of charges an additional borderline term enters the transport equation. It is caused by the Maxwell field and denoted by $(**)$ in \eqref{eq: Vlasov equation}. It eventually yields a small $\varepsilon$-growth for the energy of the distribution function similar to the borderline term $(*)$ in the same equation. To establish this behavior it is necessary that the energy controlling the Maxwell field, which appears in this term, is uniformly bounded in time. A small loss for that energy would prevent the estimates from closing.\\
However, the Maxwell equations themselves have a borderline term marked by $(*)$ in \eqref{eq: elliptic Maxwell system}. This term is caused by the presence of charged particles. It can be controlled by the energy of the distribution function but would in this case pick up its small growth and prevent the bootstrap argument from closing. At this point it is important to use the fact that the respective term in the Maxwell equations is determined by the matter current, which, in turn is determined to leading order by the energy density, which, due to the continuity equation does not have a loss in decay in comparison with the energy of the distribution function. This observation is the key that enables us to close the hierarchy of estimates and thereby the bootstrap argument.

\subsection{Organization of the paper}
In Section~\ref{Sec: preliminaries} notations, geometric setup and energies are introduced.  In Section~\ref{Sec: estimating energy-momentum}, bounds on the energy-momentum tensor in terms of the energies are given. In Section~\ref{Sec: control of the momentum support} we provide the estimate for the momentum support. Section~\ref{Sec: energy estimate} provides the energy estimates for the distribution function and the Maxwell fields. The energy density is estimated in Section \ref{Sec: energy estimates divergence identity}. Lapse and shift vector are estimated in Section~\ref{Sec: elliptic estimates for X and N}. The energy estimates for the spatial geometry are given in Section~\ref{Sec: energy estimate for geometric objects}. In Section~\ref{Sec: total energy estimate} we provide the energy estimate for the total energy incorporating all previously derived estimates. Section~\ref{Sec: global existence} provides the proof of the main theorem based on the foregoing sections.

\subsection*{Acknowledgements} 
This work was supported in part by the Austrian Science Fund (FWF) via the project \emph{Geometric transport equations and the non-vacuum Einstein flow} (P 29900-N27).

\section{Preliminaries}\label{Sec: preliminaries}
In this section we recall various notations and facts on the geometric setup from \cite{AF17,BFK19}. 
\subsection{Notation} 
In this paper, $ \mathcal{M} = \R \times M$, with $M$ being a three-dimensional compact Riemannian manifold, denotes a four-dimensional Lorentzian manifold equipped with Lorentzian metrics $\mathbf{h}$ and $h$, and the associated covariant derivatives $\bcov$ and $\nabla$, respectively. Further, Riemannian metrics on $M$ will be denoted by $\gamma$, $\tilde{g}$, $g$,  the associated Christoffel symbols by $\widehat{\Gamma}[\gamma]$, $\widetilde{\Gamma}[\tilde{g}]$, $\Gamma[g]$, and the associated covariant derivatives by $\fixcov$, $\urescov$, $D$, respectively. We will also denote the determinants of a generic metric by $|\cdot|$.
The Laplacian of $g$ is then defined as $\Delta = \trg D^2$. Moreover,   $\dvol$ will stand for the Riemannian measure induced on $M$ by $g$. The Riemannian inner products on a tangent space $T_xM$ at a point $x$ is given by $\langle \cdot , \cdot  \rangle_{g}$ and $\langle \cdot , \cdot  \rangle_{\tilde{g}}$, respectively. The Hodge-Laplacian acting on differential forms on $M$ will be denoted by $\Delta_H = d^\ast d + d d^\ast$ where $d$ and $d^\ast$ denote the exterior derivative and the codiffernetial with respect to the metric $g$ on $M$, respectively. $\mcL_Y$ will represent the Lie-derivative in the direction of a vector field $Y$. We occasionally will use the notations $\nhat := N/3 - 1$ and $\xhat := X/N$. The Greek indices will stand for the spacetime coordinates on $\mathcal{M}$ and the Latin indices will denote the coordinates on $M$, whereas the coordinates on the tangent bundle $TM$ of $M$ will be denoted by the bold Latin letters $\mathbf{a}, \mathbf{b}, \mathbf{c}, \ldots \in \{ 1, \ldots , 6 \}$. Furthermore, we denote the standard ($L^2$-based) Sobolev norm with respect to the fixed metric $\gamma$ of order $\ell \geq 0$ by $\lVert \cdot \rVert_{H^\ell (M)}$ for all functions and symmetric tensor fields on $M$. For brevity we write $H^\ell \equiv H^\ell (M)$. Throughout this paper, $C$ denotes any positive constant which is uniform in the sense that it does not depend on the solution of the system once a smallness parameter $\delta$ for the initial data and the initial time $T_0$ are chosen. Moreover, if $\delta$ is decreased or $T_0$ is increased, $C$ will keep its value. Nevertheless, the actual value of $C$ may change from line to line.

\subsection{Background geometry}

Throughout the rest of the paper, we consider the Einstein space $(M, \gamma)$ with 
$\rm{Ric}[\gamma] =  - \frac{2}{9} \gamma$. Then, the Einstein operator
$\Delta_E$ (cf.~\cite{Besse-Book} and \cite{Kr2015}) associated with $\gamma$ acting on symmetric $2$-tensors is defined by
$\Delta_E \equiv - \widehat{ \Delta}_\gamma - 2 \mathring{R}_\gamma$, 
where $\widehat{ \Delta}_\gamma := \gamma^{i j} \fixcov_i \fixcov_j $ and $( \mathring{R}_\gamma u )_{i j} := \text{Riem}[\gamma]_{i k j \ell} u^{k \ell} $ for an arbitrary symmetric $2$-tensor $u$. 
The lowest positive eigenvalue of the Einstein operator $\lambda_0$ obeys $9 \lambda_0\geq 1$ in the present setting (cf.~ \cite{AF17}), which is relevant for the energy estimate for metric and second fundamental form. In addition, $\ker \Delta_E = \{ 0 \}$ holds in the present setting (cf.~\cite{AM11} and \cite{AF17}), which assures the coercivity of the same energy. 

\subsection{Spacetime, gauges and rescaled variables}

We consider the spacetime $( \mathcal{M}, \mathbf{h})$  and write the unrescaled Lorenztian metric $\mathbf{h}$ in the ADM formalism as
$$
\mathbf{h}
=
- \widetilde{N}^2 dt \otimes dt
+
\tilde{g}_{i j}
  \big(
  dx^i + \widetilde{X}^i dt 
 \big)
 \otimes
  \big(
  dx^j + \widetilde{X}^j dt 
 \big)
\,,
$$
where $\widetilde{N}$ and $\widetilde{X}$ are the lapse function and the shift vector field. Let  $\widetilde{\Sigma}$ and $\tau$ be the trace-free part  and the trace of the second fundamental form of the hypersurfaces $\{ t = \const. \}$ which we assume that all have constant mean curvature $\tau$. 
\subsubsection{CMCSH gauge}
In the spacetime setting introduced above, the constant mean curvature spatial harmonic (CMCSH) gauge (cf. \cite{AM2003}) is achieved by
\begin{equation}
\begin{aligned}
 t & = \tau 
 \,,
 \\
 \tilde{g}^{i j}
  \left(
      \widetilde{\Gamma}[\tilde{g}]^\ell_{i j} - \widehat{\Gamma}[\gamma]^\ell_{i j}
  \right)
  &=
  0
  \,.
\end{aligned}
\end{equation}

\subsubsection{Rescaled variables}

We further define rescaled quantities $g$, $N$, $X$, and $\Sigma$ by
\begin{equation}\label{eq: rescaled geometric variables}
 g_{ij}:=\tau^2\tilde{g}_{ij}
 \,,
 \quad
g^{ij}:=\tau^{-2}\tilde{g}^{ij}
\,,
 \quad
 N:=\tau^2\widetilde{N}
 \,,
\quad
 X^i:=\tau \widetilde{X}^i
 \,, 
\quad
 \Sigma_{ij}:=\tau \widetilde \Sigma_{ij}
\,,
\end{equation}
and we introduce a rescaled time $T$ by $ \tau := \tau_0 \cdot e^{- T}$ for $ T \in (-\infty, \infty)$  and some fixed $\tau_0 < 0$.
Note that $ \p_T = - \tau \p_\tau $ and $ d\tau  =  - \tau dT$.
We occasionally will use the dot  notation to denote the $T$-derivative.

\subsubsection{Rescaled geometry}
In the new time coordinate and with the rescaled variables the Lorentzian metric reads 
\begin{equation}\label{eq: conform metric}
 \mathbf{h} 
 =
 (\tau_0)^{-2} e^{2 T}
 \left[
  - N^2 dT^2
  +
  g_{i j}
  \big(
  dx^i - X^i dT
 \big)
 \otimes
  \big(
  dx^j - X^j dT
 \big)
 \right]	
 =:
 \tau^{-2} h
 \,.
\end{equation}

\vspace{0.5em}
\begin{rem}
The Milne solution with the choice of the time coordinate $\tau$ and $\tau_0 = -3$ reads
$$
 h_{\text{{\scriptsize Milne}}}
 =
 e^{2 T}
 \left(
  - dT^2 + \tfrac{1}{9} \gamma
 \right)
 \,.
$$
\end{rem}

\vspace{0.5em}
Now, let $\Pi$ be the second fundamental form of the hypersurfaces $\{ T = \const. \}$ with respect to the metric $h$. Then, it is readily seen that
\begin{equation}
 \Pi = 
 - \Sigma 
 + 
 N^{-1}
 \big(
  1 - \tfrac{N}{3}
 \big) g
 \,.
\end{equation}
One therefore finds that the future-directed timelike unit normal $e_0$ of the hypersurfaces $\{ T = \const. \}$ with respect to $h$ is
\begin{equation}
 e_0 = N^{-1} (\p_T + X)
 \,.
\end{equation}
With this definition the following lemma holds.
\begin{lem}[\cite{BFK19}, Lemma~2.1]\label{lem: lemma 2.1 BFK}
For any $f \in C^\infty (\mathcal{M})$ and $\xi \in C^\infty (\R, \Omega^1 (M) ) $
\begin{align*}
 \mcL_{e_0} g
 &=
 - 2 \Pi
 \,,
 \\
 \left[ \mcL_{e_0}, \Delta_g   \right] f
 &=
  2 \scalpr{\Pi}{D^2 f}
  +
 \scalpr{D \log N}{D \p_{e_0} f}
 +
 \scalpr{\mathcal{S}}{D f}
 \\
 & \quad 
 +
 2 \scalpr{\Pi}{D \log N \otimes Df}
 -
 \trg \Pi \scalpr{D \log N}{D f}
 \,,
 \\
 \left[ \mcL_{e_0}, \divg   \right] \xi
 &=
 2 \scalpr{\Pi}{D \xi}
 +
 \scalpr{D \log N}{\mcL_{e_0} \xi}
 +
 \scalpr{\mathcal{S}}{D f}
\\
 & \quad 
 +
 2 \scalpr{\Pi}{D \log N \otimes \xi}
 -
 \trg \Pi \scalpr{D \log N}{\xi}
 \,,
\end{align*}
where $ \mathcal{S} := 2 \divg \Pi - D \, \trg \Pi$.
\end{lem}

Throughout this paper, the index $\nullhat$ refers to the vector field $e_0$ whereas $0$ refers to the time-function $\tau$. For the later use we summarize the connection coefficients with respect to the Lorentzian metric $h$, i.e. $ \Gamma[h]^\alpha_{\beta \gamma}$,  using the Koszul formula
\begin{equation}\label{eq: Christoffel symbols of h}
\begin{aligned}
  \Gamma[h]^{\nullhat}_{\nullhat \nullhat}   
  &= 0
  \,, 
  \\
  \Gamma[h]^{i}_{\nullhat \nullhat}  
  &=   
  g^{i j} N^{-1} \p_j N
  \,,
\end{aligned}
\quad
\begin{aligned}
  \Gamma[h]^{\nullhat}_{i \nullhat} 
   &= 0 
   \,,
  \\
   \Gamma[h]^{i}_{j \nullhat}
 &=
 -  g^{i k} \Pi_{ k j}
 \,,
\end{aligned}
\quad
\begin{aligned}
  \Gamma[h]^{\nullhat}_{i j}
 &=
 - \Pi_{i j}
 \,,
 \\
 \Gamma[h]^{i}_{ j k}
  &=
  \Gamma[g]^{i}_{ j k}
  \,.
\end{aligned}
\end{equation}

\subsubsection{Slice-adapted gauge}

On a vector potential $A \in \Omega^1 (\mathcal{M})$ the Lorenz gauge $\mathbf{d}  \star A = 0$ is usually imposed.
As discussed in \cite{BFK19}, it turns out to be difficult to work with the Lorenz gauge in the present context. Instead, we impose the \emph{slice-adapted} gauge, as introduced in \cite{BFK19}, which is adapted to the hypersurfaces of the foliation. This gauge determines $A$ uniquely by requiring that  the spatial components of $A$ associated to the foliation, denoted by the $1$-forms $\omega \in C^\infty (\R, \Omega^1 (M) ) $ with $\omega(\p_i) = A (\p_i)$, are divergence-free and orthogonal to the kernel of the Hodge Laplacian, and the component of $A$ in direction of the vector field $e_0$, denoted by $\Psi := A (e_0) \in C^\infty (\mathcal{M})$ has vanishing integral on each spatial slice. This is indeed achievable:
\begin{lem}[\cite{BFK19}, Lemma 5.4]\label{lem: slice-adapted gauge}
 Let $F \in \Omega^2 (\mathcal{M}) $ be exact. Then, there exists a unique form $A \in \Omega^1 (\mathcal{M})$ with $ \mathbf{d} A = F$  such that
\begin{equation}
 \divg \omega
 =
 0
 \,,
 \qquad
 \omega \perp \ker (\Delta_H)
 \,,
 \qquad 
 \int_M \Psi \dvol
 =
 0
 \,.
\end{equation}
\end{lem}
\begin{rem}
We consider $\omega \in C^\infty (\R, \Omega^1 (M) )$ as an element in $\Omega^1(\mathcal{M})$ by requiring $\omega(e_0) = 0$. With this choice one could write $A = \omega + \Psi e^\ast_0$ where $e^\ast_0 \in \Omega^1(\mathcal{M})$ is dual to $e_0$.  
\end{rem}

\subsection{Rescaled Einstein equations} 

In this subsection we use the $(3+1)$-dimensional ADM formalism to establish the rescaled Einstein equations.

\subsubsection{Matter quantities}

The matter quantities which appear in the ADM formulation of Einstein equations read (cf., e.g.~\cite{Rendall-Book})
\begin{equation}\label{matter-quant}
 \tilde{\rho}
 :=
 \widetilde{N}^{2} \mathbf{T}^{00}
\,,
\quad
\tilde{j}_i
:=
\widetilde{N} \mathbf{T}^0_i
\,,
\quad
\widetilde{S}_{i j}
 :=
 8 \pi  
  \left[
   \mathbf{T}_{i j} - \frac{1}{2}  \left( \tr_{\mathbf{h}} \mathbf{T} \right) \tilde{g}_{i j}
  \right]
 \,,
\end{equation}
where $ \tilde{\rho}$ and $\tilde{j}_i$ are the unrescaled energy density and the matter current, respectively.

\subsubsection{Einstein equations}

The Einstein flow in CMCSH gauge reads
\begin{subequations}\label{eq: rescaled Einstein}
\begin{align}  
 R - \absolg{\Sigma}^2 + \frac{2}{3}
 &=
 4 \tau \cdot \rho
 \,,
 \\
 D^i \Sigma_{i j}
 &=
 \tau^2 j_j
 \,,
 \\ \label{lapse}
  \left(
  \Delta_g - \tfrac{1}{3}
 \right) N
 &=
 N
 \big(
  \absolg{\Sigma}^2 + \underbrace{\tau \cdot \eta}_{(*)}
 \big)
 -
 1
 \,,
 \\
 \Delta_g X^i 
 +
 \tensor{R}{^i _j} X^j
 &=
 2 D_j N \Sigma^{i j}
 -
 D^i  \left( \tfrac{N}{3} - 1 \right)
 +
 2 N \tau^2 j^i
 \nn
 \\
 &
 \quad
 -
 \red{ 2 }
 \left(
  N \Sigma^{j k} - D^j X^k
 \right)
 \left(
   \Gamma^i_{j k} - \widehat{\Gamma}^i_{j k}
 \right)
 \,,
 \\
 \p_T g_{i j}
 &=
 2 N \Sigma_{ i j}
 +
 2 \big( \tfrac{N}{3} - 1 \big) g_{i j}
 -
 \mcL_X g_{i j}
 \,,
 \\
 \p_T \Sigma_{i j}
 &=
 -2 \Sigma_{i j}
 -
 N
 \left(
  R_{i j} + \tfrac{2}{9} g_{i j} 
 \right) 
 +
 D^2_{i j} N
 +
 2 N \Sigma_{i k} \Sigma^k_j
 \nn
 \\
 &
 \quad
 -
 \tfrac{1}{3}  \left( \tfrac{N}{3} - 1 \right) g_{i j}
 -
 \left(
  \tfrac{N}{3} - 1
 \right)
 \Sigma_{i j}
 -
 \mcL_X \Sigma_{i j}
 +
 N \tau S_{i j}
 \,.
\end{align}  
\end{subequations}
Here, we use the notations $ R_{i j}  :=  \text{Ric}[g]_{i j}$,  $R := \trg  \text{Ric}[g] $ and the rescaled matter quantities
\begin{equation}
 \rho
 :=
 4 \pi \tilde{\rho} \cdot  \tau^{-3}
 \,,
 \quad
 \eta
  :=
 4 \pi
 \left(
  \tilde{\rho} + \tilde{g}^{i j} \mathbf{T}_{i j}
 \right)
 \cdot \tau^{-3}
 \,,
 \quad
 j^i
 :=
  8 \pi \tilde{j}^i \cdot \tau^{-5}
 \,,
 \quad
 S_{i j}
 :=
 8 \pi  \widetilde{S}_{i j} \cdot  \tau^{-1}
 \,.
\end{equation}
\begin{rem}
The term $(*)$ is a borderline term, which only yields no loss on the decay for the gradient of the lapse, when the decomposition for the pressure $\eta$ below is used, which identifies the energy density as the leading order term in $\eta$.
This has already been observed in \cite{AF17} and is in particular not due to the coupling between Maxwell and Vlasov part.
\end{rem}
We decompose the rescaled energy density as
\begin{equation}
 \eta
 =
 \rho
 +
 \tau^2 \underline{\eta}
 \,,
\end{equation}
where
\begin{equation}\label{eq: etabar}
 \underline{\eta}
 =
 4 \pi \tilde{g}^{i j} \mathbf{T}_{i j} \cdot \tau^{- 5}
 \,.
\end{equation}

For the later use we also denote $T^{a b} : = \tau^{-7} \mathbf{T}^{a b}$  and we consider the following lemma which is proved in \cite{AM2003} (cf.~ also \cite{AM11}).
\begin{lem}[\cite{AM11}, Lemma~6.2]
In the CMCSH gauge the following identity holds
 \begin{equation}
  R_{i j} + \frac{2}{9} g_{i j}
  =
  \frac{1}{2} \mathcal{L}_{g , \gamma} (g - \gamma)_{i j}
  +
  \mathbb{J}_{i j}
  \,,
 \end{equation}
 where
$$
 \mathcal{L}_{g , \gamma} (g - \gamma)
 :=
 - \Delta_{g, \gamma}  (g - \gamma)
 -
 2 \mathring{R}_\gamma (g - \gamma)
 \,,
 \quad \text{with} \quad
 \Delta_{g, \gamma}  (g - \gamma)_{i j}
 :=
 \frac{1}{\sqrt{|g|}} 
 \fixcov_k
  \left[
    g^{k \ell} \sqrt{|g|} \fixcov_\ell \, ( g - \gamma )_{ i j}
  \right]
 \,,
$$ 
and
$$
  \Hnorm{\mathbb{J}}{k-1}
  \leq
  C \red{ \Hnorm{g - \gamma}{k}^2 }
 \,. 
$$
\end{lem}
An important fact is the uniform positivity of the lapse function, which is used throughout the remainder of the paper.
\begin{lem}\label{lem: uniformly positiveness}
The lapse function $N$ is uniformly positive. In particular, one has
\begin{equation}
 0 < N \leq 3
 \,.
\end{equation}
\begin{proof}
 The proof follows from applying the maximum principle to the elliptic equation for the lapse function.
\end{proof}
\end{lem}

\subsection{Vlasov matter}

In this subsection we give a quick introduction to the Vlasov matter and then rescale the momentum and finally derive the rescaled transport equation.

\subsubsection{The mass-shell relation}\label{subsub: mass-shell relations}

Throughout this paper we assume that all particles have the same positive mass $m=1$ modeled by a distribution function with the domain
\begin{equation}\label{eq: mass-shell}
 \mcP
 =
 \{
  (x, \widetilde{\mathbf{p}}) :
  \absol{ \widetilde{ \mathbf{p} } }_{\mathbf{h}}^2 = -1 \,, \tilde{p}^0 < 0 
 \}
 \subset
 T \mathcal{M}
 \,,
\end{equation}
where $\widetilde{\mathbf{p}} := \tilde{p}^\mu \p_\mu$ with $\p_0 = \p_\tau$ and $\tilde{p}^\mu$ being the canonical coordinates on the tangent bundle of $\mathcal{M}$. One can associate an energy-momentum tensor to a distribution function $\tilde{f} : \mcP \rightarrow [0, \infty ) $ by
\begin{equation}
{\vlasT}^{\alpha \beta} [ \tilde{f} ] (x)
 :=
 \int_{ \mcP_x} \tilde{f} \tilde{p}^\alpha \tilde{p}^\beta  d\mu_{\mcP_x} 
 \,,
\end{equation}
where $d\mu_{\mcP_x}$ is the Riemannian measure induced on $\mcP_x$ by the Lorentzian metric $\mathbf{h}$ at a given point $x$, and is given by
$$
 d\mu_{\mcP_x} 
 :=
 \frac{ \sqrt{ |\mathbf{h}|} }{- \tilde{p}_0} d\tilde{p}^1 \wedge d\tilde{p}^2 \wedge d\tilde{p}^3
 =
 \frac{ \widetilde{N}   }{- \tilde{p}_0} \, d\mu_{\tilde{p}}
 \,,
$$
where
$$
 d\mu_{\tilde{p}} 
 := 
  |\tilde{g}|^{\tfrac{1}{2}}  d\tilde{p}^1 \wedge d\tilde{p}^2 \wedge d\tilde{p}^3
  \,.
$$
We consider the projection map $ \text{pr}  : \mcP \rightarrow T\mathcal{M}$ which does $(t, x^i, \tilde{p}^0, \tilde{p}^i) \mapsto (t, x^i, \tilde{p}^i)$. Then, instead of using $\tilde{f}$ we deal with the function $f := \tilde{f} \circ \text{pr}^{-1} $ which we refer to as \emph{distribution function} for the remainder of the paper.

\subsubsection{The rescaled momentum}

We rescale the momentum vector field as
\begin{equation}\label{eq: scaling}
 \tilde{p}^a = \tau^2 p^a
 \,.
\end{equation}
As a result we have $\p_{\tilde{p}^i} =  \tau^{-2} \p_{p^i} $. The unrescaled mass-shell relation in \eqref{eq: mass-shell} gives (cf.,~e.g., \cite{Sarbach2014})
\begin{equation}
 \tilde{p}^0
 =
  \big(
  \widetilde{N}^2 - \absol{\widetilde{X}}_{\tilde{g}}^2
  \big)^{-1}
  \left[
   \langle \widetilde{X} , \tilde{p}  \rangle_{\tilde{g}}
   +
   \sqrt{ \langle \widetilde{X} , \tilde{p}  \rangle_{\tilde{g}}^2 + \big(  \widetilde{N}^2 - \absol{\widetilde{X}}_{\tilde{g}}^2  \big) (1 + \absol{\tilde{p}}_{\tilde{g}}^2 ) }
  \right]
  \,.
\end{equation}
Rescaling the variables in the previous equation and denoting $p^0 := \tau^{-2} \tilde{p}^0$, we find
\begin{equation}\label{eq: p0}
 p^0
 =
 N^{-1}
 \big(
  1 - \absolg{\xhat}^2
 \big)^{-1}
 \big[
 \tau \scalpr{\xhat}{p} + \phat
 \big]
 \,,
\end{equation}
or equivalently,
\begin{equation*}
 p^0
 =
 \frac{1 + \tau^2 \absolg{p}^2}{N \big[ \phat - \tau \scalpr{\xhat}{p} \big]}
 \,,
\end{equation*}
where
\begin{equation}
 \phat
 :=
 \sqrt{ \tau^2 \scalpr{\xhat}{p}^2 
   +
 \big(  1 - \absolg{\xhat}^2  \big) (1 + \tau^2 \absolg{p}^2 )  
 }
 \,.
\end{equation}
Furthermore, we find 
\begin{equation}\label{eq: p0 phat}
 \tilde{p}_0
 =
 \mathbf{h}_{0 \alpha} \tilde{p}^\alpha
 =
 - \widetilde{N} \phat
 \,,
\end{equation}
which, in particular, implies that $ d\mu_{\mcP_x}  =  \phat^{\phantom{.} -1} \, d\mu_{\tilde{p}} $. We further use the notation
$
 \underline{p} := N p^0.
 $
\begin{rem}
 Note that \eqref{eq: p0} reduces to $p^0 = \sqrt{1 + \absolg{p}^2}$ when $X = 0$ and $N = 1$. These are the values that correspond to the background geometry.
\end{rem}

\subsubsection{The transport equation}

The transport equation in the presence of the electromagnetic field reads
\begin{equation}\label{eq: Vlasov unrescaled}
 \tilde{p}^\mu \p_\mu  f
 -
 \left(
  \widetilde{\Gamma}^i_{\mu \nu} \tilde{p}^\mu \tilde{p}^\nu
  +
  q \tilde{p}^\alpha \tensor{F}{_\alpha ^ i}
 \right)
 \p_{\tilde{p}^i} f
 =
 0
 \,,
\end{equation}
where $ \widetilde{\Gamma}^\alpha_{\beta \gamma} \equiv  \widetilde{\Gamma}[\mathbf{h}]^\alpha_{\beta \gamma}$.
We wish to rewrite \eqref{eq: Vlasov unrescaled} in terms of the rescaled variables. We start with the unrescaled Christoffel symbols of $\mathbf h$ (cf.~\cite{AF17})
\begin{equation}
\begin{aligned}
  \widetilde{\Gamma}[\mathbf{h}]^a_{b c}
  &=
  \Gamma[g]^a_{b c}
  +
  N^{-1}
  \left(
   \Sigma_{b c} + \tfrac{1}{3} g_{b c}
  \right)
  X^a
 \,,
 \\ 
 \widetilde{\Gamma}[\mathbf{h}]^a_{0 0}
  &=
  \tau^{-2} \Gamma^a
  \,,
  \\
  \widetilde{\Gamma}[\mathbf{h}]^a_{0 b}
  &=
  \tau^{-1}
  \left(
   - \delta^a_b + \Gamma^a_b
  \right)
  \,,
\end{aligned} 
\end{equation}
where 
\begin{subequations}\label{eq: Gammas}
\begin{align}
 \Gamma^a
 &:=
 - \p_T X^a
 -
 X^a
 -
 \frac{2}{3} (N - 3) X^a
 +
 X^b D_b X^a
 -
 2 N \Sigma^a_b X^b
 +
 N D^a N
 \nn
 \\
 &
 \qquad
 +
 \left[
  N^{-1} \p_T N
  -
  N^{-1} X^b D_b N
  +
  N^{-1} 
   \left(
    \Sigma_{b c} + \tfrac{1}{3} g_{b c}
   \right)
   X^b X^c
 \right] X^a
 \,, \label{eq: Gamma ast}
 \\
 \Gamma^a_b
 &:=
 - N \Sigma^a_b
 +
 \frac{1}{3} \delta^a_b (3 - N) 
 +
 D_b X^a
 -
 N^{-1} X^a D_b N
 +
 N^{-1} 
   \left(
    \Sigma_{b c} + \tfrac{1}{3} g_{b c}
   \right)
   X^c X^a
  \,. 
\end{align}  
\end{subequations}
We use notations $\Gamma^\ast$ and $\Gamma^\ast_\ast$ when we want to suppress the indices of the above two objects. The rescaled transport equation finally takes the following form, using the natural horizontal and vertical derivatives on $TM$, $\A_a :=\p_a - p^i \Gamma^k_{a i} \B_k$ and $\B_a := \p_{p^a}$,

\begin{equation}\label{eq: Vlasov equation}
\begin{split}
 \p_T f
 =&
 \tau N   \frac{p^a}{\pbar} \A_a f
 -
 \underbrace{\tau^{-1} \frac{\pbar}{N} \Gamma^a \B_a f}_{(*)}
 +
 2 p^a \B_a f
 -
 2 p^c \Gamma^a_c  \B_a f
\\
&
 -
 \tau
  \left(
   \Sigma_{b c}
   +
   \tfrac{1}{3} g_{b c}
  \right)
  X^a \frac{p^b p^c}{\pbar} \B_a f
  +
  \underbrace{\tau q \F^a \B_a f}_{(**)}
 \,,
\end{split}
\end{equation}
where
\begin{equation}\label{eq: eq for frakF}
 \F^i
 :=
 h^{i j}
  F_{0 j}
  +
  \frac{p^a}{p^0}  
  \left(
     h^{i j}   F_{a j}
   +
   \frac{\tau}{N} F_{a 0} \xhat^i
 \right)
 \,,
\end{equation}
with $h^{i j} = g^{i j} - N^{-2} X^i X^j$, which can be read off from the metric $h$.

\begin{rem}
The terms marked by $(*)$ and $(**)$ are borderline terms. In particular, the term marked by $(**)$ originates from the Maxwell field and the slow decay is caused by the first term in \eqref{eq: eq for frakF}. In combination with the factor $\tau$ the term $F_{0j}$ appears in the energy for the Maxwell field \eqref{eq: L2 norm of F}, which itself is $\varepsilon$-small but does not decay. As a consequence the term $(**)$ in the Vlasov equation yields a growth for the energy of the distribution function of the type $\exp(C\varepsilon T)$, however, only if we obtain sharp estimates for the energy of $F$. Therefore, we need to avoid that the loss for the energy of the distribution function couples back into the equation for the Maxwell field.
\end{rem}

\subsection{Maxwell equation}

We start with the the Maxwell equation which in Heaviside--Lorentz units reads
\begin{equation}\label{eq: maxwell}
 \hb^{\lambda \mu} \bcov_\lambda F_{\mu \alpha}
 =  -  q \int f \tilde{p}_\alpha \,  d\mu_{\mcP_x} =:{J}_\alpha.
\end{equation}

Recall that the index $0$ refers to the vector field $\p_\tau$.
Rescaling the result according to \eqref{eq: rescaled geometric variables} and \eqref{eq: scaling} yields
\red{
\begin{equation*}
{J}_\alpha
 =
 - q  \int f \frac{ \tilde{p}_\alpha }{\widehat{p}}  \, d\mu_{\tilde{p}}
 =
 - q \tau^3  
 \int f  \frac{\tilde{p}_\alpha}{\widehat{p}}   \,  d\mu_p
 \,,
\end{equation*}
}
where we used 
\begin{equation*}
 d\mu_{\tilde{p}}
 =
 \tau^{3} |g|^{\frac{1}{2}}  dp^1 \wedge dp^2 \wedge dp^3 
 =:
 \tau^{3} \, d\mu_p
 \,.
\end{equation*}
Thus, by \eqref{eq: p0 phat} we have 
\begin{equation}\label{eq: J null}
{J}_0
 =
  \tau q N
 \int  f  \,     d\mu_p
 =:
 \tau \jnull
 \,,
\end{equation}
and using
\begin{equation}
 \frac{\tilde{p}_a}{\widehat{p}}
 =
  \frac{\tau^{-1} X_a p^0 + g_{ab} p^b}{ \widehat{p}}
 =:
 - N \mathcal{P}_a
\end{equation}
we find
\begin{equation}\label{eq: J vec}
{J}_i
 =
 q N \tau^3 
 \int  f \, \mathcal{P}_i  \,   \,  d\mu_p
 =:
  \tau^3 \jvec_i
 \,.
\end{equation}
\red{
The definition of $\mathcal{P}_a$ is motivated by the relation $ \B_a p^0 = -  \tau^2 \mathcal{P}_a$ (cf.~Appendix~\ref{App: momentum derivatives}).
}
Keeping \eqref{eq: conform metric} in mind, for the left-hand side of \eqref{eq: maxwell} we get 
\begin{equation*}
 \hb^{\lambda \mu} \bcov_\lambda F_{\mu \alpha}
 =
 \hb^{\lambda \mu}
\left[
 \nabla_\lambda F_{\mu \alpha}
 -
 \left(
 \tensor{\widetilde{\Gamma}}{^\beta_{\mu\lambda}}
 -
 \tensor{\Gamma}{^\beta_{\mu\lambda}}
 \right) 
  F_{\beta \alpha}
 -
  \left(
 \tensor{\widetilde{\Gamma}}{^\beta_{\alpha \lambda}}
 -
 \tensor{\Gamma}{^\beta_{\alpha \lambda}}
 \right)  
 F_{\mu \beta}
\right]
=
 \hb^{\lambda \mu}
 \nabla_\lambda F_{\mu \alpha}
=
\tau^2 h^{\lambda \mu}
 \nabla_\lambda F_{\mu \alpha}
\,,
\end{equation*}
\red{
where we used the facts from Appendix~\ref{App: Diff Tensor}.
}
Hence,
\begin{equation}\label{eq: rescaled maxwell}
 h^{\lambda \mu}
 \nabla_\lambda F_{\mu \alpha}
 =
\tau^{-2} \tilde{J}_\alpha
 \,.
\end{equation}
We compute the following for  Proposition below, which gives the rescaled Maxwell equations
\begin{equation}
\begin{aligned}
 F_{0 i}
 &=
 - \tau^{-1}
 \left[
  \p_T \omega_i 
  + 
  \p_i \left( N \Psi - X^j \omega_j \right) 
 \right]
 \,,
 \\
 F_{\nullhat i}
 &=
 N^{-1}
 \left(
  \p_T \omega_i + X^j \p_j \omega_i
 \right)
 -
 \p_i \Psi
 \,,
 \\
 F_{i j}
 &=
 \p_i \omega_j - \p_j \omega_i
 \,,
\end{aligned}
\end{equation}
and
$$
 \tau^{-2} \tilde{J}_{\nullhat}
 =
 \tau^{-2} \tilde{J} (e_0)=
\tau^{-2}
 \left(
 N^{-1}
  \tilde{J}_T + \widehat{X}^j \tilde{J}_j
 \right) =
\tau^{-2}
 \left(
  - \tau N^{-1} \tilde{J}_0 
  +
  \widehat{X}^j \tilde{J}_j
 \right) =
 - N^{-1} \jnull
 +
 \tau \widehat{X}^j \jvec_j
 \,,
$$
where we used $\tilde{J}_T = \tilde{J} (\p_T) = \tilde{J}( - \tau \p_\tau) = -\tau \tilde{J}_0 $.

\begin{prop}\label{prop-EMS}
 Let $F \in \Omega^2(\mathcal{M})$ be exact which solves \eqref{eq: rescaled maxwell} and let $A \in \Omega^1(\mathcal{M})$ be a vector potential for $F$ which satisfies the slice-adapted gauge conditions in Lemma~\ref{lem: slice-adapted gauge} with the same $\Psi$ and $\omega$ given there. Then, we have
\begin{subequations}\label{eq: elliptic Maxwell system}
\begin{align}
 \Delta_g \Psi
 &=
 - \divg (\Psi \cdot D\log(N))
 -
 [\mcL_{e_0}, \divg]\omega
  - 
   \underbrace{N^{-1} \jnull}_{(*)}
   +
   \tau \widehat{X}^j \jvec_j  
 \,,\label{eq: elliptic of psi}
 \\
 (\mcL_{e_0} (\mcL_{e_0} \omega))_k
 +
  \Delta_H \omega_k 
 &=
 \p_k( \p_{e_0} \Psi )
 +
 \p_{e_0} \Psi \cdot \frac{\p_k N}{N}
 +
 \Psi \cdot  \p_k( \p_{e_0} \log(N) )
 \nn
 \\
 &
 \quad
 +
 \red{
 g^{i j} \frac{\p_i N}{N} F_{j k}
 }
 +
 \trg \Pi \cdot F_{\nullhat k}
 +
\red{2} g^{ij} \Pi_{i k} F_{j \nullhat}
 -
 \tau \jvec_k 
 \,. \label{eq: wave eq. of omega}
\end{align}
\end{subequations}
\begin{rem}
The term $(*)$ is borderline in the following sense. As defined in \eqref{eq: J null} it is given up to a factor by the integral of the distribution function. A straightforward estimate by the energy of the distribution function would induce a small growth for the field $\Psi$, which would prevent the estimates from closing. To obtain a uniform bound on $\Psi$ it is crucial to observe that the leading order term in $(*)$ is in fact the energy density, which is established in Lemma \ref{lem: J estimates}.

\end{rem}

\begin{proof}[Proof of Proposition \ref{prop-EMS}]
From \eqref{eq: maxwell} and using the slice-adapted gauge, we get (cf.~Eqs.~(5.16)--(5.18) in \cite{BFK19})
\begin{align}
g^{ij} D_i F_{j \nullhat}
 &=
 \Delta_g \Psi
 +
 \divg (\Psi \cdot D\log(N))
 +
 [\mcL_{e_0}, \divg] \omega
 \,,
 \\
 g^{ij} D_i F_{j k}
 &=
 - \Delta_H \omega_k 
 +
 \trg \Pi \cdot F_{\nullhat k}
 +
 g^{ij} \Pi_{i k} F_{j \nullhat}
 \,,
 \\
 D_{\nullhat} F_{\nullhat k}
 &=
 (\mcL_{e_0} (\mcL_{e_0} \omega))_k
 -
 \p_k( \p_{e_0} \Psi )
 -
 \p_{e_0} \Psi \cdot \frac{\p_k N}{N}
 -
 \Psi \cdot  \p_k( \p_{e_0} \log(N) ) 
 \nn
 \\
 &
\quad
  -
 g^{i j} \frac{\p_i N}{N} F_{j k}
 +
 g^{ij} \Pi_{ i k} F_{\nullhat j} 
 \,.
\end{align}
Inserting the results into \eqref{eq: rescaled maxwell} and using \eqref{eq: J null} and \eqref{eq: J vec} finishes the proof.
\end{proof}
\end{prop}
\subsection{Local existence}
A local-existence-theory for the system \eqref{eq: rescaled Einstein}, \eqref{eq: Vlasov equation} and \eqref{eq: elliptic Maxwell system} can be derived based on the ideas of \cite{BFK19,F16}.

\begin{prop}
Consider CMC-initial data 
\begin{equation}(g_0,\Sigma_0,f_0,\omega_0,\dot\omega_0)\in \mathscr{B}_{\delta}^{6,5,5,6,5}(\gamma,0,0,0,0)
\end{equation}
with $\delta>0$ sufficiently small.
Let $I \subset \R$ be a compact interval with $T_0 \in I$. Then, there exists a $T>0$ and a unique solution $(g,k,f,\omega,\dot\omega)$
to the Einstein--Vlasov--Maxwell system with $J = (T_0 - T, T_0 + T)$ launched by this initial data. $T$ depends continuously on the $H^6 (M)$-, $H^5(M)$-, $H_{\rm{Vl}, 5,4,c}(TM)$-, $H^6 (M)$-, and $H^5(M)$-norm of ${g}_0$, ${\Sigma}_0$, ${f}_0$, ${\omega}_0$, and ${\dot\omega}_0$, respectively. The following regularity properties hold
\begin{align*}
 (g, k), (\omega,\dot \omega) 
 & \in 
 C_b (J, H^6 \times H^5) \cap C_b^1 (J, H^5 \times H^4)
 \,,
 \\
 f 
 &\in 
 C_b (J, H_{\rm{Vl}, 5,4,c}(TM))
 \,.
\end{align*}
Moreover, the solution is either global in time, i.e. $T=\infty$ or 
$$
\limsup _{t\nearrow T}\left(\|g-\gamma\|_{H^6}+\|\Sigma\|_{H^5}+\|\omega\|_{H^6}+\|\dot\omega\|_{H^5}+\||f\||_{5,4}\right)\geq 2\delta. 
$$

\end{prop}

\begin{proof}
The local existence theory can be established similar to the corresponding theorems for the Vlasov case in \cite{AF17} and the Maxwell case in \cite{BFK19}. When considering the Maxwell equations in Lorenz gauge we obtain a wave-type equation for the 4-potential in the form
$$
 \square_{H, h} A_\alpha = \tau^{-2}{J}_\alpha
 \,,
$$
where $\square_{H, h}$ denotes the Hodge wave operator of the metric $h$.

Complementing the Einstein-Vlasov system with Maxwell terms by this equations does not change the structure of the elliptic-hyperbolic system with respect to the Einstein-Vlasov case considered in \cite{F16}. The additional wave-type equation can be treated at the same order of regularity as the evolution equations for the spatial metric g (which is decomposed into first order equations). 
A local-existence theorem analogue to Theorem 4.2
of \cite{F16} follows by a similar proof, where the vector potential is controlled in the same regularity class as the metric. In consequence we obtain a local-existence theory for the EVMS in CMCSH-Lorenz gauge.\\
Then, as shown in \cite{BFK19}, we perform a gauge-transformation for the vector potential of the Maxwell field to obtain a potential obeying the slice-adapted gauge while conserving the regularity of the potential and thereby of the solution. We conclude at this point that we have existence and uniqueness of a local-in-time solution of the EVMS in CMCSH-slice-adapted gauge in the respective regularity class.\\
It remains to prove the continuation criterion. For this purpose we need to ensure that a solution, which is small in CMCSH-slice-adapted gauge remains small when its vector potential is transformed to the Lorenz gauge. We therefore consider the gauge transformation of the vector potential
\begin{equation}
A_\mu=A'_\mu+\partial_\mu \Lambda,
\end{equation}
such that $A'_\mu$ fulfils the Lorenz gauge. Then $\Lambda$ solves the wave equation
\begin{equation}
 \Box_h\Lambda
 =
 -
 \p_{e_0} \Psi
 -
  N^{-1} (1 - h^{i j} h_{i j}) \left(  \p_{e_0} \Lambda - \Psi \right)
 -
 \xhat^i \xhat^j D_i \omega_j 
 +
 h^{i j} \Pi_{i j} \Psi
 +
 N^{-1} g^{i j} \p_j \omega_i
 \,,
\end{equation}
with trivial initial data.  Here, we used \eqref{eq: Christoffel symbols of h} and relations in Appendix~\ref{App: Diff Tensor}. From this follows that smallness of the solution in slice-adapted gauge in the considered regularity class implies smallness of the solution in Lorenz gauge in the corresponding regularity class on short time intervals, as $\Lambda$ is determined by the right-hand side of the wave equation. Note therefore that $\partial_{e_0}\Psi$ fulfils the elliptic equation \eqref{eq: elliptic of psi}, which improves its regularity by ellipticity of the equation. This implies the continuation criterion in slice-adapted gauge.
\end{proof}
\subsection{Coupling of the equations and consequences for the long-time behavior}
Having derived all equations we give a brief summary on the coupling issues in the remark below.
\begin{rem}
 The rescaled EVMS consists of systems of equations \eqref{eq: rescaled Einstein}, \eqref{eq: Vlasov equation}, and \eqref{eq: elliptic Maxwell system}. In this rescaled system we have particular matter-electromagnetic coupling terms, i.e., $\tau q \F^i \B_i f$ and $ N^{-1} \jnull $, which appear as source terms in the Maxwell equations and in the Vlasov equations, respectively. Those terms are not present in the respective individual cases, i.e., pure electromagnetism in \cite{BFK19} or pure uncharged matter in \cite{AF17}. In addition, they constitute principal terms in the sense of their decay properties. Both terms decay slower or at the same rate as the slowest decaying neighbouring source terms and thereby constitute potential obstacles when analyzing the rescaled equations. The major observation of the following bootstrap analysis is that those terms can indeed be handled in the existing bootstrap hierarchy of the previous works and in turn yield the same asymptotic behaviour as the respective cases.

\end{rem}

\subsection{Energy-momentum tensor}
We compute and collect all relevant terms from the energy-momentum tensor in the following.
The unrescaled energy-momentum tensor is given by
$$
 \mathbf{T}
 =
{\vlasT}
 +
{\maxT}
 \,.
$$
Accordingly, we have for the matter variables defined in \eqref{matter-quant},
$$
 \tilde{\rho}
% &=
% \tilde{\rho}_{\text{{\tiny{Vlasov}}}}
% +
% \tilde{\rho}_{\text{\tiny{Maxwell}}}
 \equiv
  \tilde{\rho}_{\text{{\tiny{V}}}}
 +
 \tilde{\rho}_{\text{\tiny{M}}}
 \,,
 \qquad
  \tilde{j}^a
% &=
% \tilde{j}^a_{\text{{\tiny{Vlasov}}}}
% +
% \tilde{j}^a_{\text{\tiny{Maxwell}}}
\equiv
  \tilde{j}^a_{\text{{\tiny{V}}}}
 +
 \tilde{j}^a_{\text{\tiny{M}}}
 \,,
 \qquad
 \underline{\eta}
% &=
%  \underline{\eta}^{\text{{\tiny{Vlasov}}}}
%  +
%  \underline{\eta}^{\text{\tiny{Maxwell}}}
\equiv
  \underline{\eta}^{\text{{\tiny{V}}}}
  +
  \underline{\eta}^{\text{\tiny{M}}}
  \,,
$$
where scripts V and M stand for Vlasov and Maxwell matter fields, respectively.

Then, the rescaled Vlasov matter quantities take the following form
\begin{subequations}\label{eq: Vlasov energy momentum tensor}
\begin{align}
 \rho_{\text{{\tiny{V}}}}
 &=
 4 \pi N^2 
 \int f \frac{(p^0)^2}{\phat}  \,  d\mu_p
 \,,
 \\
 j^a_{\text{{\tiny{V}}}} 
 &=
  8 \pi N
 \int f \frac{p^0 p^a}{\phat}  \,  d\mu_p
 \,,
 \\
 \underline{\eta}^{\text{{\tiny{V}}}}
 &=
 4 \pi 
 \int f \frac{ \absolg{ p + \tau^{-1} p^0 X }^2 }{\phat}  \,  d\mu_p
 \,,
 \\
 ^\mathrm{V}{\!}T^{a b}
 &=
 8 \pi
 \int f \frac{p^a p^b}{\phat}   \,  d\mu_p
 \,.
\end{align}
\end{subequations}

For the Maxwell field we start with the unrescaled energy-momentum tensor (again in the Heaviside--Lorentz units)
\begin{equation}
{\maxT}_{\mu \nu}
 =
  \tensor{F}{_\mu ^\alpha}  F_{\nu \alpha}
 -
  \frac{1}{4}
  \hb_{\mu \nu} F_{\alpha \beta} F^{\alpha \beta}
 \,.
\end{equation}

We first compute $F_{\alpha \beta} F^{\alpha \beta}$ in terms of the rescaled variables
\begin{equation}
 F_{\alpha \beta} F^{\alpha \beta}
 =
 - 2 \tau^6 N^{-2} g^{i j} F_{0 i} F_{0 j}
 +
 2\tau^5 N^{-2} X^i g^{j k} F_{0 k} F_{i j}
 +
 \tau^4 h^{i k} h^{j \ell} F_{i j} F_{k \ell}
 \,.
\end{equation}
Thus, in terms of rescaled variables ${\maxT}_{00}$ takes the form
\begin{equation}\label{eq: T00}
{\maxT}_{00}
 =
 \tau^2 h^{ij} F_{0 i} F_{0 j} 
 +
 \frac{\tau^{-4}}{4}   
 \left(
    1 -  \absolg{\xhat}^2
 \right) N^2
 F_{\alpha \beta} F^{\alpha \beta}
 \,.
\end{equation}
For the  off-diagonal components of the energy-momentum tensor we have
\begin{equation}\label{eq: T0i}
{\maxT}_{0 i}
 =
 \tau^2 h^{j k} F_{0 k} F_{i j}
 -
 \frac{1}{4} \tau^{-3} X_i \, F_{\alpha \beta} F^{\alpha \beta} 
 \,.
\end{equation}
Finally, the spatial part reads
\begin{equation}\label{eq: MaxTij}
{\maxT}_{i j}
 =
 - \tau^4 N^{-2} F_{0 i} F_{0 j}
 -
 \tau^3 N^{-2} X^k 
 \left(
  F_{0 i} F_{j k} + F_{0 j} F_{ik}
 \right)
 +
 \tau^2 h^{ k \ell} F_{i k} F_{ j \ell}
 -
 \frac{1}{4} \tau^{-2}
 F_{\alpha \beta} F^{\alpha \beta}
 \,.
\end{equation}
Then,  the rescaled Maxwell quantities can be expressed by the components of the energy-momentum tensor of Maxwell field computed above as
\begin{subequations}
\begin{align}
 \rho_{\mathrm{M}}
 &=
 4 \pi \tau N^{-2}
 \left[
{\maxT}_{00} 
  -
  2 \tau^{-1} X^i  ( {\maxT}_{0 i} )
  +
  \tau^{-2} X^i X^j ({\maxT}_{i j})
 \right]
 \,,
 \\  
 j_{\mathrm{M}}^i
 &=
 8 \pi N^{-1}
 \left[
  N^{-2} X^i ( {\maxT}_{00} )
  +
  \tau^{-1} 
   \left(
    \xhat^i \xhat^j + h^{i j}
   \right)
   {\maxT}_{0 j}
  +
  \tau^{-2} h^{i k} X^j ({\maxT}_{k j} )
 \right]
 \,,
 \\
 \underline{\eta}^{\text{{\tiny{M}}}}
 &=
 4 \pi \tau{-3} g^{i j} ({\maxT}_{i j})
 \,.
\end{align}
\end{subequations}

\subsection{Norms for matter fields}\label{Sec: Norms}
We recall in the following briefly the norms used to control distribution function and Faraday tensor as introduced in \cite{F17} and \cite{BFK19}, respectively.

\subsubsection{$L^2$-Sobolev energy of the distribution function}

$L^2$-Sobolev energies of the distribution function can be defined based on the Sasaki metrics with respect to $\gamma$ and $g$.  For $\gamma$, $ \bm{\gamma}
 \equiv
 \gamma_{i j} dx^i \otimes dx^j
 +
 \gamma_{i j} \fixcov p^i \otimes \fixcov p^j$, 
where $\fixcov p^i := dp^i + \widehat{\Gamma}^i_{j k} dp^j dp^k$, defines a metric on $TM$. We denote the associated covariant derivative by $\widehat{\mathbf{D}}$. The volume form induced by $\bm{\gamma}$ is then given by
$
 d\mu_{\bm{\gamma}}
 :=
 - \absol{\gamma} \prod_{i = 1}^{3} dx^i \wedge dp^i
 \,.
$
A weighted metric on $TM$ is defined by
$
 \underline{\bm{\gamma}}
 :=
 \gamma_{i j} dx^i \otimes dx^j
 +
 \langle p \rangle_\gamma^{-2}  \gamma_{i j} \fixcov p^i \otimes \fixcov p^j
 \,,
$
where $\langle p \rangle_\gamma := \sqrt{1 + \absol{p}^2_\gamma}$. We then define the $L^2$-Sobolev energy of the distribution function with respect to the associated Sasaki metric of $\gamma$	by
\begin{equation}
 \vertiii{f}_{\ell , \mu}
 :=
  \sqrt{\sum_{k \leq \ell} \int_{TM} \langle p \rangle_\gamma^{2\mu + 4 (\ell - k)} \absol{\widehat{\mathbf{D}}^k f }^2_ {\underline{\bm{\gamma}}} \, d \mu_{ \bm{\gamma}} }
\end{equation}
The function space associated with the above norm is denoted by $H_{\text{Vl}, \ell, \mu} (TM)$. Moreover, pointwise estimates are taken with respect to 	the following $L^\infty_x L^2_p$-norm
$$
 \vertiii{f}_{\infty, \ell , \mu}
 :=
 \sup_{x \in M}
  \left\{
   \sqrt{ \sum_{k \leq \ell}  \int_{T_xM} \langle p \rangle_\gamma^{2\mu + 4 (\ell - k)} \absol{\widehat{\mathbf{D}}^k f }^2_ {\underline{\bm{\gamma}}} \, d \widehat{\mu}_p }
  \right\}
  \,,
$$
where $d \widehat{\mu}_p := |\gamma|^{\tfrac{1}{2}} dp^1 \wedge dp^2 \wedge dp^3 $.  Then, the following lemma holds.
\begin{lem}
For $f$ sufficiently regular, there exists a constant $C$ 	such that
$$
 \vertiii{f}_{\infty, \ell , \mu}
 \leq
 C
 \vertiii{f}_{\ell + 2 , \mu}
$$
holds.
\end{lem}

In addition, we consider the $L^2$-Sobolev energy associated with the dynamical metric $g$. The Sasaki metric with respect to $g$ is
$
 \mathbf{g}
 \equiv
  g_{i j} dx^i \otimes dx^j
 +
 g_{i j} Dp^i \otimes Dp^j
 \,, 
$
where $ D p^i := dp^i + \Gamma^i_{j k} dp^j dp^k$. The associated covariant derivative and the connection coefficients are denoted by $\mathbf{D}$ and $\bm{\Gamma}$ (cf.~Appendix~\ref{App: connection coefficients of Sasaki}), respectively, and the volume form induced by $\mathbf{g}$ on $TM$ is given by
$
 d\mu_{\mathbf{g}}
 :=
 - \absol{g} \prod_{i = 1}^{3} dx^i \wedge dp^i
 \,.
$
Then, we define analogously a weighted version of the Sasaki metric associated with $g$ by
$
 \underline{\mathbf{g}}
 :=
 g_{i j} dx^i \otimes dx^j
 +
 \langle p \rangle^{-2}   g_{i j} Dp^i \otimes Dp^j
 \,, 
$
where $\langle p \rangle := \sqrt{1 + \absolg{p}^2} $. We finally define the $L^2$-Sobolev energy of the distribution function with respect to the Sasaki metric associated with the dynamical metric $g$ by
\begin{equation}\label{eq: L2 energy of f}
 \mcE_{\ell, \mu}(f)
 :=
 \sqrt{\sum_{k \leq \ell} \int_{TM} \langle p \rangle^{2\mu + 4 (\ell - k)} \absolsasag{\sasacov^k f }^2 \, d \mu_{\mathbf{g}} }
\end{equation}
Under suitable smallness assumptions, $ \mcE_{5, 4}(f)$ and $ \vertiii{f}_{5, 4}$ are equivalent.

\subsubsection{$L^2$-Sobolev norm of the Faraday tensor}

In the following we define the $L^2$-Sobolev norm of the Faraday tensor which measures the perturbation of the Maxwell field. It is defined as follows (cf.~\cite{BFK19})
\begin{align}\label{eq: L2 norm of F}
 \Fnorm{F}{\ell}^2
&:=
 \sum_{ k \leq \ell} \int_M
 \Big(
  \tau^2 g^{i j} D_{i_1} \cdots D_{i_ k} F_{0 i} \,
  D^{i_1} \cdots D^{i_ k} F_{0 j} 
\nn
 \\
 &
 \qquad\qquad\quad
 + 
 g^{i j} g^{a b} 
  D_{i_1} \cdots D_{i_ k} F_{i a} \,
  D^{i_1} \cdots D^{i_ k} F_{j b}
 \Big)
 d\mu_g
 \,.
\end{align}
The factor $\tau^2$ in the first term compensates the the growth of the $\tau$-component coming from $A_\tau = \tfrac{dT}{d\tau} A_T = - \tau^{-1} A_T $, relative to the spatial components of $F$. 

\subsection{Smallness}
To determine the final estimates we use a standard bootstrap argument. We define a set of smallness conditions for the dynamical quantities which serves as the bootstrap assumptions. We define 
\begin{align}
 \mcB^{6,5,5,6,5}_{\delta, \tau} (\gamma, 0, 0, 0, 0)
 :=
 \Big\{
 & (g, \Sigma, f, \omega, \dot{\omega}) \in H^6 \times H^5 \times  H_{\text{Vl}, 5, 4} \times H^6 \times H^5 \Big|
 \nn
 \\
\qquad \quad
 &
  |\tau|^{- \frac{1}{2}} 
  \left(
   \Hnorm{ g - \gamma}{6}
   +
   \Hnorm{\Sigma}{5}
  \right)
  +
   \Hnorm{\rho}{4}
   +
   |\tau|^\frac{1}{2} \cdot
     \vertiii{f}_{ 5,4}
   +
    \Fnorm{F}{5}
  <
  \delta 
 \Big\}
 \,. \label{eq: smallness assumption}
\end{align}
We say that $(g(\tau), \Sigma(\tau), f(\tau), \omega(\tau), \dot{\omega}(\tau))$ is \emph{$\delta$-small} if $(g, \Sigma, f, \omega, \dot{\omega}) \in \mcB^{6,5,5,6,5}_{\delta, \tau} (\gamma,0,0, 0, 0)$. To refer to $\delta$-small data we also use the term \emph{smallness assumptions}. A direct consequence of the smallness assumption is stated in the following lemma.
\begin{lem}
For $\delta$-small data \eqref{eq: smallness assumption} with $\delta$ sufficiently small, the energies $\energy{5}{4}$ and $\vertiii{f}_{5, 4}$ are equivalent.
\begin{proof}
The proof follows straightforwardly from the definitions and the smallness assumption.
\end{proof}
\end{lem}

The smallness assumptions imply the smallness of the perturbation for the lapse function and shift vector. The following result is a corollary of 	 Proposition~\ref{prop: estimates for N and X}.
\begin{cor}
 For any $\delta > 0$ there exists a $\bar{\delta} > 0$ such that  $ (g, \Sigma, f, \omega, \dot{\omega}) 
 \in  
 \mcB^{6,5,5,6,5}_{\delta, \tau} (\gamma, 0, 0, 0, 0)$ implies
\begin{equation}
 |\tau|^{-1} 
 \left(
  \Hnorm{N - 3}{6}
  +
  \Hnorm{X}{6}
 \right)
 <
 \bar{\delta}
 \,.
\end{equation}
\end{cor}

\section{Estimating the energy-momentum tensor}\label{Sec: estimating energy-momentum}

In this section we will estimate the components of the energy-momentum tensor of the Maxwell field and other matter quantities by the norms defined in the previous section.
\begin{lem}\label{lem: estimate of TMaxwell}
For \red{$k > 3/2$}
%\Hbar{I think this is wrong in the BFK paper. It should be strictly greater than 3/2}
 we have the following estimate
\begin{equation}
 \Hnorm{{\maxT}_{00}}{k}
 +
 |\tau|^{-1} \cdot \Hnorm{{\maxT}_{0i}}{k}
 +
 |\tau|^{-2} \cdot \Hnorm{{\underline{\maxT}}}{k}
 \leq
 C
% \left(
%  1 + \Hnorm{\xhat}{k}^2
% \right)
 \Fnorm{F}{k}^2
 \,,
\end{equation}
where ${}^{\rm M}{\underline{\mathbf{T}}}$ denotes the spatial part of ${\maxT}$ and $ C = C
 \left(
   \Hnorm{N}{k}, \Lnorm{N}{\infty}, \Lnorm{N^{-1}}{\infty}, \Hnorm{X}{k}
 \right)
$.
%\Hbar{should we mention this every time?}
%
\begin{proof}
The statement of the lemma follows directly from the expressions \eqref{eq: T00}--\eqref{eq: MaxTij} (cf.~also Lemma 4.2 in \cite{BFK19}).
\end{proof}
\end{lem}

We summarize the estimates of the rescaled matter quantities in the following proposition.   
\begin{prop}\label{Prop: matter estimates}
For \red{$k > 3/2$} the following estimates hold
\begin{equation}
\begin{aligned}
 \Hnorm{\rho}{k}
  & \leq
  C
%\left(
%  1 + \Hnorm{\xhat}{k}^2
% \right)
 \left[
  |\tau| \cdot  \Fnorm{F}{k}^2
   +
  \energy{k}{3}
 \right]
 \,,
 \\
 \Hnorm{j}{k}
 &\leq
 C
%\left(
%  1 + \Hnorm{\xhat}{k}^2
% \right)
  \left[
   \Fnorm{F}{k}^2
   +
  \energy{k}{3}
 \right]
 \,,
 \\
 \Hnorm{\underline{\eta}}{k}
 &\leq
 C
% \left(
%  1 + \Hnorm{\xhat}{k}^2
% \right)
  \left[
   |\tau|^{-1} \cdot \Fnorm{F}{k}^2
   +
  \energy{k}{4}
 \right]
 \,,
 \\
 \Hnorm{S}{k}
 & \leq
 C
% \left(
%  1 + \Hnorm{\xhat}{k}^2
% \right)
 \left[
  |\tau| \cdot \Fnorm{F}{k}^2
  + 
  \tau^2 \energy{k}{4}
  +
  \Hnorm{\rho}{k} 
 \right]
 \,,
\end{aligned}
\end{equation}
where $C = C \left( \Hnorm{N}{k}, \Lnorm{N}{\infty}, \Lnorm{N^{-1}}{\infty} , \Hnorm{X}{k} \right)$.
\begin{proof}
Combining  Lemma~13 in \cite{AF17} with   Lemma~\ref{lem: estimate of TMaxwell} results in the estimates of the proposition. 
\end{proof}
\end{prop}

\subsection{Pointwise estimates on the momentum variables}

In this subsection we provide some useful auxiliary estimates used further below.
\begin{lem}\label{lem: momentum inequalities}
The following estimates hold for large $T$, using the smallness of the lapse function and shift vector, and provided that the momenta have compact support 
 \begin{align}
  \frac{\absolg{p}}{\phat}
  &\less
  |\tau|^{-1}
  \,,
  \quad
  \text{or}
  \quad
  \blue{
    \frac{\absolg{p}}{\phat}\less \absolg{p}
  }
  \,,
  \\
   \frac{p^0}{\phat}
  &\less
  \frac{1}{N \big( 1 -\absolg{\xhat}^2 \big)}
  \,,
  \\
 \absolg{\mathcal{P}}
  &\less
  |\tau|^{-1} \cdot \absolg{\xhat}
  +
  N^{-1} \absolg{p}
  \,.
 \end{align}
\end{lem}
\begin{proof}
The proof follows straightforwardly from the explicit expressions for the respective quantities.
\end{proof}

\section{Control of the momentum support}\label{Sec: control of the momentum support}
Based on the characteristic system associated to the transport equation, we derive an estimate on the radius of the momentum support of the distribution function.

The characteristic system associated with the rescaled transport  equation reads
\begin{align}
 \frac{d x^a}{dT}
 &=
 - \tau \frac{p^a}{p^0}
 \,,
 \\
 \frac{dp^a}{dT}
 &=
 \tau^{-1} \Gamma^a p^0
 -
 2 p^a
 +
 2 \Gamma^a_b \, p^b
 +
 \tau
 \left[
  \Gamma^a_{b c}
  +
  N^{-1}
  \left(
   \Sigma_{b c} + \tfrac{1}{3} g_{b c}
  \right)
  X^a
 \right]
 \frac{p^b p^c}{p^0}
 \red{ 
  -
  \tau q \F^a 
 } 
 \,.
\end{align}
With the definition of the auxiliary quantity (cf.~\cite{AF17})
\begin{equation}
 \mathbf{G}(T, x, p)
 := 
 \absolg{p}^2
 \,,
\end{equation}
one finds its derivative along a characteristic 
\begin{equation}
 \frac{d \mathbf{G}}{dT}
 =
 |p|^2_{\dot{g}}
 +
 2 \tau^{-1} \scalpr{\Gamma^\ast}{p} \, p^0
 +
 4 g_{i k}  \Gamma^i_j \, p^j p^k
 +
 \frac{2 \tau}{N} 
 \left(
  \Sigma_{b c} + \tfrac{1}{3} g_{b c} 
 \right)
 \frac{p^b p^c}{p^0} \scalpr{X}{p}
\red{ -
 4 \tau q \scalpr{\F}{p} 
 } 
 \,,
\end{equation}
which implies the following result.

\begin{lem}\label{lem: estimate of dG to dT}
\red{For $\delta$-small data \eqref{eq: smallness assumption} with $\delta$ sufficiently small}, the following estimate holds for any characteristic
\begin{equation}
 \left| \frac{d \mathbf{G}}{d T} \right|
  \leq 
 C
  \left[
   \absolg{\dot{g}}
   +
   \absolg{\Gamma^\ast_\ast}
   +
   \left(
     1+  \absolg{\Sigma}  
   \right) \absolg{X}
  \right]
  \mathbf{G}
  +
  C
  \left(
  |\tau|^{-1} \cdot \absolg{\Gamma^\ast} 
  +
  \red{
   |q \tau| \cdot \absolg{\F}
  }
  \right)
  \sqrt{\mathbf{G}}
 \,.
\end{equation}
\end{lem}

Defining the supremum of the values of $\mathbf{G}$ in the support of $f$ at a fixed time $T$ analogous to \cite{AF17} as
\begin{equation}
 \mcG[T]
 :=
 \sup
 \left\{
  \sqrt{\mathbf{G}(T,x,p)} \, | \,
  (x,p) \in \supp f(T, .\, , . )
 \right\}
 \,,
\end{equation}
we derive an estimate to control the momentum support in the following proposition.
\begin{prop}\label{Prop: curly g}
\red{For $\delta$-small data \eqref{eq: smallness assumption} with $\delta$ sufficiently small}, the following estimate holds 
\begin{align}
 \left. \mcG  \right|_{T}
& \leq
  \left[
     \left. \mcG  \right|_{T_0}
     +
   C \int_{T_0}^T
    \left(
      e^s \Hnorm{\Gamma^\ast}{2} 
      +
      \red{
       \Fnorm{F}{2}
     }
    \right) 
      ds
  \right]
 \nn
 \\
 &
 \quad
 \times 
 \exp
  \left[
    C \int_{T_0}^T
    \left(
      \Hnorm{\Sigma}{2}
      +
      \Hnorm{N-3}{2}
      +
      \Hnorm{X}{2}
      +
      \Hnorm{\Gamma^\ast_\ast}{2}
      +
      \red{
      |q| e^{-s} \Fnorm{F}{2}
     }
    \right) 
    ds
  \right]    
 \,.
\end{align}
\begin{proof}
From Lemma \ref{lem: estimate of dG to dT} and the definition of $\F^i$, we get
\begin{equation*}
\begin{aligned}
\frac{d}{dT} \sqrt{\mathbf{G}}
&\leq
C
\left(
 \Hnorm{\Sigma}{2}
 +
 \Hnorm{N-3}{2}
 +
 \Hnorm{X}{2}
 +
 \Hnorm{\Gamma^\ast_\ast}{2}
 +
 \red{
 |q \tau| \cdot \Fnorm{F}{2}
 }
\right)
\sqrt{\mathbf{G}}
\\
& \quad
+
C 
 \left(
  |\tau|^{-1}  \Hnorm{\Gamma^\ast}{2}
  +
  \red{
  \Fnorm{F}{2}
  }
 \right)
\,.
\end{aligned}
\end{equation*}
Applying Gr\"onwall's lemma completes the proof.
\end{proof}
\end{prop}

\section{Energy estimates}\label{Sec: energy estimate}

\subsection{Energy estimate for the distribution function}

We first give the energy estimate for the distribution function.

Before we proceed, we prove the following lemma:
\begin{lem}\label{lem: inequality for phi}
Assume that \red{$\delta$-small data \eqref{eq: smallness assumption} with $\delta$ sufficiently small} holds. Then, for a generic sufficiently regular function $\phi$ on $TM$, the following inequality holds
\begin{equation}
 \sum_{m \leq k} \int \absolg{ D^m \int \phi \,  d\mu_p}^2 \, \volg
 \less
 \sum_{m \leq k} \int
 \left(
  \int \bar{p} \absolg{ \sasacov^m \phi} d\mu_p
 \right)^2
\volg
 \,.
\end{equation}
\begin{proof}
To obtain this inequality, we use (cf.~(4.26) in \cite{AF17})
\begin{equation}
 D_a \int \phi \, d\mu_p
 =
 \int \A_a \phi \, d\mu_p
 \,.
\end{equation}
For higher covariant derivatives one gets terms of the form
\begin{equation}\label{eq: mixed terms}
 p^i \bullet \text{Riem}[g] \bullet \sasacov^\ell \phi
 \,,
\end{equation}
where $\ell \leq m-1$ and $\bullet$ denotes various abstract contractions of tensor indices. Then, the smallness assumption on the metric yields
\begin{align}
 \sum_{m \leq k} 
 \int \left| D^m \int \phi d\mu_p \right|^2 d\mu_g
 &\less 
 \sum_{m \leq k}  \int
 \left(
  \int \absolg{ \sasacov^m \phi} \, d\mu_p
  +
   \int \absolg{p} \absolg{ \sasacov^{m-1} \phi} \, d\mu_p 
 \right)^2
\volg
\nn
\\
 &\less 
 \sum_{m \leq k}  \int
 \left[
  \int 
  \left( 1 + \absolg{p} \right)  
  \absolg{ \sasacov^m \phi} \, d\mu_p
 \right]^2
 \volg
\nn
\\
& \less 
 \sum_{m \leq k} \int
 \left(
  \int \bar{p} \absolg{ \sasacov^m \phi} \, d\mu_p
 \right)^2
\volg
\,.
\end{align}
\end{proof}
\end{lem}

\begin{prop}\label{prop: T derivative of L2-energy of f}
Let $f$ be a solution of the transport equation \eqref{eq: Vlasov equation}. 
For $\delta > 0$ sufficiently small and $(g, \Sigma, f, \omega, \dot{\omega}) \in \mcB^{6,5,5,6,5}_{\delta, \tau} (\gamma,0,0, 0, 0)$ the estimate
\begin{align}
\p_T \energysq{k}{\mu}
&\leq
C
 \Big[
   \Hnorm{N - 3}{k}
   +
   \Hnorm{\Sigma}{k}
   +
   \Hnorm{X}{k + 1}
   +
   |\tau| \cdot \Hnorm{N^{-1} \p_T X}{k}
   +
  |\tau|^{-1} \cdot \Hnorm{N^{-1} \Gamma^\ast}{k}
\nn
 \\
 &
 \qquad
  +
  \Hnorm{\Gamma^\ast_\ast}{k}
  +
  \Hnorm{(\Sigma + g) X}{k}
  +
  |\tau| \mcG
  +
 \blue{
   |q | \cdot \Fnorm{F}{k}
  }   
 \Big]
 \energysq{k}{\mu}
 \,,\label{eq: T derivative of L2-energy of f}
\end{align}
holds for $k > 5/2$ and $\mu \geq 3$, provided that $|\tau| \mcG$ is bounded by a constant.
\begin{proof}
Recall the following notation. Using the introduced bold Latin indices we define the frame $\{\theta_{\mathbf{a}}\}_{\mathbf{a} \leq 6} := \{ \A_1, \A_2, \A_3, \B_1, \B_2, \B_3 \}$. That is, the indices $\mathbf{a} \in \{1,2,3 \}$ correspond to the horizontal directions whereas the indices $\mathbf{a} \in \{4, 5, 6\}$ correspond to the vertical directions.
The proof is identical to the proof of  Proposition 12 in \cite{AF17} except for the last term. However, for the last term which comes from the Maxwell part of the transport equation \eqref{eq: Vlasov equation} one can proceed similarly. Indeed; we consider 
\begin{equation}\label{eq: Sasaki cov on dt f}
 \sasacov_{\mathbf{a}_1} \cdots  \sasacov_{\mathbf{a}_k} \p_T f
 \,,
\end{equation}
which arises from $\p_T \absolsasag{\sasacov^k f }^2 $ when applying $\p_T$ to the $L^2$-energy \eqref{eq: L2 energy of f}. For all the other terms involving $\p_T$ we refer to Proposition 12 in \cite{AF17}. Therefore, we will focus on the Maxwell part of the transport equation in \eqref{eq: Sasaki cov on dt f}
\begin{align*}
 \sasacov_{\mathbf{a}_1} \cdots  \sasacov_{\mathbf{a}_k}
 \left(
    \tau q \F^a \B_a f
 \right)
 &=
 \tau q
 \Big(
  \theta_{\mathbf{a}_1} \sasacov_{\mathbf{a}_2} \cdots  \sasacov_{\mathbf{a}_k}
  -
  \sum_{2 \leq j \leq k}
   \mathbf{\Gamma}^\mathbf{c}_{\mathbf{a}_j \mathbf{a}_1} \sasacov_{\mathbf{a}_2} \cdots \sasacov_{\mathbf{a}_c} \cdots  \sasacov_{\mathbf{a}_k}
 \Big)
  \left(
     \F^a \B_a f
 \right)
 \\
 &=
  \tau q  \theta_{\mathbf{a}_1} \cdots  \theta_{\mathbf{a}_k} 
  \left(
     \F^a \B_a f
 \right)
 +
 \ldots
 +
 \tau q (-1)^{k - 1} \mathbf{\Gamma}^{\mathbf{c}_1}_{\mathbf{a}_k \mathbf{a}_1}   \mathbf{\Gamma}^{\mathbf{c}_2}_{\mathbf{c}_1 \mathbf{a}_2} 
 \cdots
  \mathbf{\Gamma}^{\mathbf{c}_{k-1}}_{\mathbf{c}_{k-2} \mathbf{a}_{k-1}} 
   \left(
     \F^a \B_a f
 \right) 
 \,,
\end{align*}
where we suppressed the mixed terms. After commuting the operator $\F^a \B_a$ to the front by using the relations in Appendix~\ref{App: Commutation} one obtains
\begin{equation}\label{eq: F B on Sasaki}
 \tau q   \F^a \B_a ( \sasacov_{\mathbf{a}_1} \cdots  \sasacov_{\mathbf{a}_k} f )
\,,
\end{equation}
and other mixed terms which are of the form (cf.~\eqref{eq: mixed terms})
\begin{equation}
  \tau q  D^{k_1} \F^a
  \left(
   p^i D^{k_2} \text{Riem}
  \right)^{k_3} 
  D^{k_4} f
 \,,
\end{equation}
with $\sum_i k_i = k$. From \eqref{eq: F B on Sasaki} and after integration by parts when evaluating  
$$
 \tau q  \sum_{k \leq \ell} \int_{TM} \langle p \rangle^{2\mu + 4 (\ell - k)} \mathbf{g}^{\mathbf{a}_1 \mathbf{b}_1} \cdots \mathbf{g}^{\mathbf{a}_k \mathbf{b}_k}
 \sasacov_{\mathbf{b}_1} \cdots  \sasacov_{\mathbf{b}_k} f \cdot \sasacov_{\mathbf{a}_1} \cdots  \sasacov_{\mathbf{a}_k} \p_T f
 \, d \mu_{\mathbf{g}}
\,, 
$$
one finds the corresponding term $|q| \cdot \Fnorm{F}{k}$ in \eqref{eq: T derivative of L2-energy of f}. All the mixed terms are of lower order and can be absorbed in the latter term.

\end{proof}
\end{prop}

\subsection{Estimate for the current density}
As argued previously, the specific form of the effective terms arising from the current, $\jnull$ and $\jvec_i$ is essential to obtain sufficient bounds on the Maxwell fields. The following lemma establishes these bounds.
\begin{lem}\label{lem: J estimates}
The following estimates 
\begin{align}
 \Hnorm{\jnull}{k}
 &\leq
 C |q| \cdot
% \left(
%  1 + |\tau| \mcG 
% \right)^3
 \Hnorm{\rho_{\text{{\tiny{V}}}} }{k}
 \,,\label{J-est}
 \\ 
 \Hnorm{\jvec}{k}
 &\leq
 C |q|
% \left(
%  1 + |\tau| \mcG
% \right)
 \left(
  |\tau|^{-1} \Hnorm{X}{k} \Hnorm{\rho_{\text{{\tiny{V}}}} }{k}
  +
  \Hnorm{j_{\text{{\tiny{V}}}} }{k}
 \right)
 \,,
\end{align}
hold, provided that $|\tau| \mcG$ is bounded by a constant and \red{$\delta$-small data \eqref{eq: smallness assumption} with $\delta$ sufficiently small} holds.
\begin{proof}
 To obtain the estimates above we need following inequalities
$$
 \sup_{ (x, p) \, \in \, \supp  f(T, .  , . )  } \phat
 \leq
 C \left( 1 + |\tau| \mcG \right)
 \,,
\quad
    \sup_{ (x, p) \, \in \, \supp  f(T, .  , . )  } \frac{\phat}{ \pbar^2} 
\leq
C
\left( 1 + |\tau| \mcG \right)^3
\,,
$$
and
$$
 \sup_{ (x, p) \, \in \, \supp  f(T, .  , . )  } \frac{1}{\pbar}
 \leq
 C \left( 1 + |\tau| \mcG \right)
 \,,
$$
where we applied the smallness conditions. Then, 
\begin{align}
 \Hnorm{ \jnull }{k}
 &=
 \Hnormadj{ q N \int f  d\mu_p }{k}
 \nn
 \\
 &\leq 
 C |q| \left( 1 + |\tau| \mcG \right)^3 \Hnorm{N^{-1} \rho_{\text{{\tiny{V}}}} }{k}
 \,,
\end{align}
from which the first estimate follows. Note here that the expression $\widehat p$ and $p^0$ appear via $\int f d\mu_p=\int f (p^0)^2/\widehat p \cdot \widehat p/(p^0)^{2} d\mu_p$ to generate $\rho_V$. When the factor $ \widehat p/(p^0)^{2}$ is hit by derivatives of the norm, the generated terms can be estimated by a uniform constant invoking the smallness assumptions. All those terms are absorbed into the constant C. 

For the second estimate we have
\begin{align}
 \Hnorm{\jvec}{k}
 &=
 \Hnormadj{q \tau^{-1} X_i \int f \frac{p^0}{\phat} d\mu_p + q \int f \frac{p_i}{\phat} d\mu_p }{k}
 \nn
 \\
 &\leq 
 C |q| \left( 1 + |\tau| \mcG \right)
 \left(
  |\tau|^{-1}  \Hnorm{X}{k}  \Hnorm{\rho_\text{{\tiny{V}}}}{k}
  +
  \Hnorm{j_\text{{\tiny{V}}}}{k}
 \right)
 \,,
\end{align}
{where $\rho_V$ and $j_V$ are defined in \eqref{eq: Vlasov energy momentum tensor}.
}
Then, the assumption on the boundedness of $|\tau| \mcG$ finishes the proof.
\end{proof}
\end{lem}
{
\begin{rem}
Estimate \eqref{J-est} enables us to estimate $\jnull$ by the energy density of the distribution function and eventually by the full energy density, which, as shown below, is uniformly bounded. This prevents this term to pick up the loss from the Vlasov energy.
\end{rem}
}

\subsection{Energy estimate for the vector potential}

As explained in Section~5.1 of  \cite{BFK19}, the structure of the system \eqref{eq: elliptic of psi}--\eqref{eq: wave eq. of omega} suggests the following energy
\begin{equation}\label{eq: energy of omega}
\begin{aligned}
 \EMax_k(\omega)
 &:=
 \sum_{\ell = 0}^{k-1} \int_M
  \scalpr{(\Delta_H)^\ell \mcL_{e_0} \omega}{\mcL_{e_0} \omega}
  +
  \scalpr{(\Delta_H)^{\ell+1} \omega}{\omega} \, d\mu_g
 \\
 &\simeq
 \Hgnorm{\mcL_{e_0} \omega}{k-1}^2
 +
  \Hgnorm{\omega}{k}^2
\end{aligned}
\end{equation}
with $k \geq 1$. In the following we omit the argument $\omega$ and write $\EMax_k \equiv \EMax_k(\omega)$ for simplicity.
\begin{rem}
The gauge condition $\omega \perp \text{ker} (\Delta_H)$ allows us to control the $L^2$-norm of $\omega$ by $\EMax_k$.
\end{rem}

\begin{rem}
Under smallness assumptions the $L^2$-Sobolev norm of the Faraday tensor \eqref{eq: L2 norm of F} is equivalent to the energy \eqref{eq: energy of omega} and the norm of $\Psi$ in the following sense 
\begin{equation}\label{eq: equivalence of F and E}
 \Fnorm{F}{k}
 \simeq
   \Hnorm{\Psi}{k+1}
   +
   \sqrt{\EMax_{k+1}}
 \,.
\end{equation}
\begin{rem}
Both terms on the right-hand side, as shown below, are uniformly bounded.
\end{rem}

\end{rem}
\begin{lem}\label{lem: dT of Ek}
Let $F$ be a solution of \eqref{eq: rescaled maxwell} and let $A \in \Omega^1(\mathcal{M})$ be a gauged vector potential for $F$, $\Psi$, and $\omega$ as in Lemma~\ref{lem: slice-adapted gauge}. Then, for $k>5/2$ and assuming $\Hnorm{N}{k}$ is bounded by some constant, we have
\begin{align}
 \p_T \EMax_k
 &\leq
 C
 \left(
  \Hnorm{D \log N}{k-1}
  +
  \Hnorm{ \Pi}{k-1}  
  +
  \Hnorm{ \mathcal{S} }{k-1}
 \right) \EMax_k
\nn
\\
&\quad
+
C 
\left[
 \Hnorm{\p_{e_0} \Psi}{k}
 +
 \left(
  \Hnorm{\p_{e_0} \log N}{k}
  +
  \Hnorm{ \Pi}{k-1}    
 \right)
 \Hnorm{\Psi}{k}
 +
  \blue{ 
   |\tau| \cdot \Hnorm{\jvec}{k-1}
   }
\right]
\sqrt{\EMax_k}
\,.\label{eq: time derivative of Ek}
\end{align}
\begin{proof}
 The lemma is proved along the same lines as the proof of  Lemma~5.7 in \cite{BFK19}, and by using \eqref{eq: wave eq. of omega} and (cf. Eq.~(5.37) in \cite{BFK19})
\begin{equation}\label{eq: interior pr F}
 \Hnorm{i_{e_0}F}{k-1}
 \leq
 C
 \left(
  \Hnorm{D \Psi}{k-1}
  +
  \Hnorm{\Psi}{k-1} \Hnorm{D \log N}{k-1}
  +
  \sqrt{\EMax_k}
 \right)
 \,,
\end{equation}
where $i_Y: \Omega^\ell (\mathcal{M}) \longrightarrow \Omega^{\ell - 1} (\mathcal{M}) $ is the interior product for any vector field $Y \in \mathfrak{X} (\mathcal{M})$. Here, by abuse of notation, we consider the Sobolev norm of $i_{e_0} F \in C^\infty (\R, \Omega^1(M))$.
\end{proof}
\end{lem}

\begin{lem} \label{lem-psi}
As long as $\Hnorm{D \log N}{k}$ is small enough the following estimate holds
\begin{align}
 \Hnorm{\Psi}{k+1}
 &\leq
 C
 \Big[
  \Hnorm{\Pi}{k-1}
  +
  \left( 1 + \Hnorm{\Pi}{k-1} \right) \Hnorm{D \log N}{k-1}
  +
  \Hnorm{ \mathcal{S} }{k-1}
 \Big]
 \sqrt{\EMax_k}
 \nn
 \\
 &\quad
 +
 \blue{
 C 
 \Big(
 \underbrace{ \Hnorm{\jnull}{k-1}}_{(*)}
 +
 |\tau| \cdot \Hnorm{X}{k-1} \Hnorm{\jvec}{k-1}
 \Big)
 }
\,.
\end{align}
\begin{rem}
The term $(*)$ is borderline, as it would yield a growth for the norm of $\Psi$. Uniform boundedness is achieved due to the estimates using the energy density as outlined above.
\end{rem}
\begin{proof}[Proof of Lemma \ref{lem-psi}]
By elliptic regularity and the fact that \eqref{eq: elliptic of psi} has a unique solution, we find
\begin{align*}
 \Hnorm{\Psi}{k+1}
 &\leq
 C 
\Big(
 \Hnorm{\Psi}{k} \Hnorm{D \log N}{k}
 +
 \Hnorm{[\mcL_{e_0}, \divg] \omega}{k-1}
 \nn
 \\
 &\qquad
 +
 \Hnorm{N^{-1}}{k - 1} \Hnorm{\jnull}{k-1}
 +
 |\tau| \Hnorm{\xhat}{k-1} \Hnorm{\jvec}{k-1}
\Big)\,.
\end{align*}
And from  Lemma~\ref{lem: lemma 2.1 BFK} it follows
\begin{equation*}
 \Hnorm{[\mcL_{e_0}, \divg] \omega}{k-1}
 \leq
 C
 \left[
  \Hnorm{\Pi}{k-1}
  +
  \left( 1 + \Hnorm{\Pi}{k-1} \right) \Hnorm{D \log N}{k-1}
  +
  \Hnorm{ \mathcal{S} }{k-1}
 \right]
 \sqrt{\EMax_k}
 \,.
\end{equation*}
Combining the results and using Lemma~\ref{lem: J estimates} completes the proof.
\end{proof}
\end{lem}

We need to estimate the term $ \Hnorm{\p_{e_0} \Psi}{k}$ appearing in \eqref{eq: time derivative of Ek}. This is done in  Lemma~\ref{lem: de0 of Psi} which need further preparations:
\begin{lem}\label{lem: frakF}
Let the integer $k > 7/2$ and assume that $\Hnorm{N}{k-2}$ is bounded. Then, we have
\begin{align}\label{eq: F0i}
 \Hnorm{i_{\p_\tau} F}{k-2}
 &\leq
 C |\tau|^{-1}
  \Hnorm{N}{k-2}
   \left(
  \Hnorm{\Psi}{k-1}
  +
   \Hnorm{\Psi}{k-2} \Hnorm{D \log N}{k-2}
  +
  \sqrt{\EMax_{k-1}}
 \right)
\nn
 \\
 &\quad
 +
C  |\tau|^{-1}
 \Hnorm{X}{k-2} \sqrt{\EMax_{k-1}} 
 \,.
\end{align}
\begin{proof}
We recall that  $i_{\p_\tau} F \in C^\infty (\R, \Omega^1(M))$. Hence,
\begin{equation}
 i_{\p_\tau} F(\p_a)
=
 F_{a 0}
 =
 - \tau^{-1}  N \cdot i_{e_0} F(\p_a) 
 +
 \tau^{-1}  X^i F_{a i}
\,.
\end{equation}
The first term can be estimated by \eqref{eq: interior pr F} and for the second term we use \eqref{eq: equivalence of F and E}. Then, invoking the assumption on $\Hnorm{N}{k-2}$ finishes the proof.
\end{proof}
\end{lem}
\begin{lem}\label{lem: e0 of J0} 
Let $k, \mu \in \N$ such that $k > 7/2$
and $\mu \geq 3$. Further, assume that $\Hnorm{N}{k-2}$ is bounded. Then, the following estimate holds
\begin{equation}
 \Hnorm{\pnull \jnull}{k-2}
 \leq
 C
 \left(
  \Hnorm{\p_T N}{k-2} 
   +
  \Hnorm{\nhat }{k-2}  
  +
  \Hnorm{X}{k-1}
 \right) 
 \Hnorm{\jnull}{k-2}
 +
 C |q|  \energy{k-1}{\mu}
 \,, 
\end{equation}
provided that $\Hnorm{\Gamma^\ast}{k-2}$ is bounded.
\begin{proof}
 We have
\begin{equation}\label{eq: time derivative of J0}
 \pnull \jnull
 =
  N^{-2} \p_T N \jnull
  +
 q \int \p_T f d\mu_p
  +
 q \int f \p_T ( d\mu_p )
 +
 \xhat^k \p_k \jnull
 \,.
\end{equation}
The third term can be rewritten as 
\begin{equation*}
 q \int f \p_T ( d\mu_p )
 =
 \blue{
  N^{-1}
  \left(
     3 \nhat - D_i X^i + 6
  \right)   
  }
  \jnull
 \,,
\end{equation*}
where we used the evolution equation for $|g|^{\frac{1}{2}}$ and the rescaling of the momentum according to \eqref{eq: scaling} to get
\begin{equation}\label{eq: dT of dp}
 \p_T ( d\mu_p )
 =
 \left(
  3  \nhat
  -
  D_i X^i
  +
  6
 \right)
 d\mu_p
 \,.
\end{equation}
For the second term in \eqref{eq: time derivative of J0}, using the transport equation, we find
\begin{align}
 \int \p_T f \, d\mu_p
 &=
 \int
 \big[
  \tau N \frac{p^i}{\pbar} \A_if 
  -
  \tau^{-1} \frac{\pbar}{N} \Gamma^a \B_a f
  +
  2 p^a \B_a f
\nn
\\
&  
\qquad\quad
  -
  2 \Gamma^a_c p^c \B_a f
  -
  \tau 
  \left(
   \Sigma_{b c} + \tfrac{1}{3} g_{b c}
  \right)
  X^a \frac{p^b p^c }{\pbar	} \B_a f
  -
 \tau q \F^i \B_i f
 \big]
 \, d\mu_p
  \,.\label{eq: integral of transport}
\end{align}
The last term after integration by parts gives
\begin{align}\label{eq: Maxwell term in Vl. eq.}
 - \tau q 
 \int 
 \F^i \B_i f \, d\mu_p
 &=
 \tau q \int \B_i \F^i \, f \, d\mu_p
\nn
\\ 
 &=
% q \tau^3
%  \int
%   \frac{p^a}{(p^0)^2} \mathcal{P}_i
%  \left[
%   F_{a j} 
%  \left( g^{i j} - \xhat^i \xhat^j \right)
%   +
%   \frac{\tau}{N} F_{a 0} \xhat^i  
%  \right]
%   f \, d\mu_p
%\nn
% \\
% &   
% +
% q \tau^2 
%  \int
%   \frac{1}{\pbar} F_{i 0 } \xhat^i \, f \, d\mu_p
\blue{
 \tau q
 \int f \, 
  \left[
    \left(
    g^{i j} - \xhat^i \xhat^j
    \right)
    F_{a j}
    +
    \frac{\tau}{N} F_{a 0} \xhat^i
  \right]
  \B_i 
  \left(   
    \frac{p^a}{p^0}
  \right)  
  \, d\mu_p
 } 
 \,.
 %\label{eq: i.b.p. F Bf}
\end{align}
The third term of the \eqref{eq: integral of transport} cancels the third term of \eqref{eq:  dT of dp}, after integration by parts. Hence,
\begin{align}
 q \int \p_T f \, d\mu_p
  +
 q \int f \p_T ( d\mu_p )
 &=
 q \int
 \big[
  \tau N \frac{p^i}{\pbar} \A_if 
  -
  \tau^{-1} \frac{\pbar}{N} \Gamma^a \B_a f
  -
  2 \Gamma^a_c p^c \B_a f
\nn
 \\
 & 
\qquad\quad 
 -
  \tau 
  \left(
   \Sigma_{b c} + \tfrac{1}{3} g_{b c}
  \right)
  X^a \frac{p^b p^c }{\pbar	} \B_a f
  +
 \tau  q \B_i \F^i  f
 \big]
 \, d\mu_p
\nn
 \\
 &=:
 q 
 ( I_1 + I_2 + I_3 + I_4 + I_5 )
  \,.
\end{align}
Invoking Lemmas~\ref{lem: momentum inequalities} and \ref{lem: inequality for phi} and the relations in Appendix~\ref{App: momentum derivatives}, integrating by parts whenever possible, and imposing the smallness assumptions, leads to
\begin{equation}
\begin{aligned}
 \Hnorm{I_1}{k-2}
 &\leq
 \blue{
C  \cdot |\tau|  \cdot   \energy{k-1}{\mu}
  }
 \,,
 \\
 \Hnorm{I_2}{k-2}
 &\leq
 \blue{
 C \cdot |\tau|  \cdot  \Hnorm{\Gamma^\ast}{k-2} \cdot
  \energy{k-2}{\mu}
  }
 \,,
 \\
  \Hnorm{I_3}{k-2}
 &\leq
 C  \cdot \Hnorm{\Gamma^\ast_\ast}{k-2} \cdot
  \energy{k-2}{\mu}
 \,,
 \\ 
 \Hnorm{I_4}{k-2}
 &\leq
 \blue{
 C \cdot |\tau| \cdot  \Hnorm{(\Sigma + g)X}{k-2}  \cdot
 \energy{k-2}{\mu + 1}
 }
 \\
\Hnorm{I_5}{k-2}
 & \leq
 \blue{
 C \cdot |\tau q| \cdot
 \Fnorm{F}{k-2} \cdot   \energy{k-2}{\mu}
 }
 \,.
\end{aligned}
\end{equation}
\blue{Note that $\energy{k-2}{ \mu + 1} \leq \energy{k-1}{\mu}$.} 
Hence,
\blue{
\begin{align*}
 \Hnorm{\pnull \jnull}{k-2}
 &\leq
 C
  \left(
    \Hnorm{N^{-2}}{k-2} \cdot \Hnorm{\p_T N}{k-2}
    +
    \Hnorm{\nhat}{k-2}
    +
    \Hnorm{X}{k-1}
   \right)
   \Hnorm{\jnull}{k-2} 
 \nn
 \\
 &  
 \quad
    +
  |q|   \sum_{i=1}^5 \Hnorm{I_i}{k-2}
\,.
\end{align*}
}
Assuming that $\Hnorm{N}{k-2}$ and \red{$\Hnorm{\Gamma^\ast}{k-2}$} are bounded, we arrive at the estimate of the lemma.
\end{proof}
\end{lem}
\begin{lem}
Let the integers $k> 7/2$ and let $ \mu \geq 3$ and assume that $\Hnorm{N}{k-2}$ is bounded. Then,
%\Hbar{the third line could be written in a more compact fashion, assuming that $|\tau|^{-1} \Hnorm{\p_T X}{k}$ and $\Hnorm{\p_T N}{k}$ are bounded }
%
\begin{align}
 \Hnorm{\pnull \jvec}{k-2}
 &\leq
 C
 \left(
  \Hnorm{\p_T N}{k-2} 
   +
  \Hnorm{\nhat }{k-2}  
  +
  \Hnorm{X}{k-1}
 \right) 
 \Hnorm{\jvec}{k-2}
 \nn
 \\
 &
 \quad
 +
 C |q|
  \Big[
      1
      +
      |\tau|
      +
    |\tau|^{-1} \cdot \Hnorm{\Gamma^\ast}{k-2}
     +
     \Hnorm{\Gamma^\ast_\ast}{k-2}
     +
     \Hnorm{\Sigma}{k-2}
     +
     \Hnorm{\nhat}{k-2}
      +
   \Hnorm{X}{k-1}  
  \nn
  \\
  &
  \qquad\quad
  +
  |\tau|^{-1}
   \left(
    \Hnorm{X}{k-2}
    +
    \Hnorm{\p_T X}{k-2}
   \right)
   +
   \Hnorm{\p_T N}{k-2}
     +
     |q|  \cdot  \Fnorm{F}{k-2}
   \Big]
   \energy{k-1}{\mu}
  \,. 
\end{align}
\begin{proof}
 We have
\begin{equation}
 \p_{e_0} \jvec_k
 =
 N^{-2} \p_T N \jvec_k
 +
 q
  \int \p_T f \mathcal{P}_k \, d\mu_p
  +
  q
  \int  f \p_T \mathcal{P}_k \, d\mu_p   
  + 
  q
  \int  f  \mathcal{P}_k \p_T ( d\mu_p)    
 +
 \xhat^i \p_i \jvec_k
 \,.\label{eq: e0 of Jk}
\end{equation}
Again, using \eqref{eq: dT of dp}, the fourth term can be rewritten as
\begin{equation} 
  q
  \int  f  \mathcal{P}_k \, \p_T  (d\mu_p)     
 = 
 \blue{
 N^{-1}
 \left(
  3 \nhat 
  -
  D_i X^i
  +
  6
 \right) }
  \jvec_k
 \,.
\end{equation}
For the second and the fifth terms of \eqref{eq: e0 of Jk} we repeat a similar calculation to the one in Lemma \ref{lem: e0 of J0}. The corresponding term which after integrating by parts cancels the fifth term of \eqref{eq: e0 of Jk} is $2  \int \mathcal{P}_k \, p^a \B_a f \, d\mu_p$. Thus, 
\begin{equation}
 2 q  \int \mathcal{P}_k \, p^a \B_a f \, d\mu_p
 +
 6 N^{-1} \jvec_k 
 =
 - 2 q  \int \B_a \mathcal{P}_k \, p^a  f \, d\mu_p
 \,.
\end{equation}
Using \eqref{eq: B p0}, under smallness assumptions, we find
\blue{
\begin{equation}
 \Hnormadj{q  \int \B_a \mathcal{P}_k \, p^a  f \, d\mu_p}{k-2}
 \leq
 C |q| 
 \energy{k-2}{\mu}
 \,.
\end{equation}
}
For the Maxwell part of the second term of \eqref{eq: e0 of Jk} one gets, after integration by parts,
\blue{
\begin{equation}
 \Hnormadj{\tau q \int \F^i \B_i f \mathcal{P}_k \, d\mu_p}{k-2}
 \leq
 C |q| 
 \left( 1 + |\tau| \right) 
 \Fnorm{F}{k-2} \cdot \energy{k-2}{\mu}
 \,.
\end{equation}
}
Therefore,
\blue{
\begin{align}
 \Hnormadj{ q  \int \p_T f \mathcal{P}_k \, d\mu_p}{k-2}
 & \leq
  C  |q|
  \Big[
    |\tau| 
    +
    |\tau|^{-1}  \left( 1 + \tau^2 \right) \Hnorm{\Gamma^\ast}{k-2}
    +
    \Hnorm{\Gamma^\ast_\ast}{k-2}
 \nn
 \\
  &
  \qquad\quad  
    +
   |\tau| \cdot  \Hnorm{(\Sigma + g)X}{k-2} 
    +
  |q| \cdot
 \Fnorm{F}{k-2} 
  \Big]
  \energy{k-1}{\mu}
  \,.
\end{align}
}
Finally, putting \eqref{eq: T derivative of mathcal P} into the third term of \eqref{eq: e0 of Jk} and making use of the relations in Appendix~\ref{app: T derivatives of momentum functions}, one finds
\blue{
\begin{align}
 \Hnormadj{q  \int  f  \p_T \mathcal{P}_k  \, d\mu_p }{k-2}
 &\leq
 C |q| 
  \Big[
    |\tau|^{-1}
    \left(
     \Hnorm{X}{k-2}
     +
     \Hnorm{\p_T X}{k-2}
    \right)
    +
    \Hnorm{\p_T N}{k-2}
 \nn
 \\
 &   
 \qquad\quad
    +
    \Hnorm{\Sigma}{k-2}
    +
    \Hnorm{\nhat}{k-2}
  \Big]
  \energy{k-2}{\mu}
  \,.
\end{align}
}
Inserting the results above and invoking the smallness assumptions yields the proof.
\end{proof}
\end{lem}
\begin{lem}\label{lem: de0 of Psi}
Let integers $k > 7/2$
and and $\mu \geq 3$. Then, as long as $\Hnorm{D \log N}{k}$ is small enough and $\Hnorm{\Pi}{k-2}$ is bounded, the following holds
\begin{align}
  \Hnorm{\p_{e_0} \Psi}{k}
 &\leq
 C
  \bigg(
  \Hnorm{ \Pi}{k-2}  
  +
  \Hnorm{ \mathcal{S} }{k-2}
  +
  \Hnorm{D \log N}{k-2}
  +
  \Hnorm{D \pnull \log N}{k-2}
 \nn
 \\
 &
 \qquad\quad
 + 
  \Hnorm{\mcL_{e_0} \mathcal{S} }{k-2}
  +
  \Hnorm{\mcL_{e_0} \Pi }{k-2}  
 \bigg)
 \left(
  \sqrt{\EMax_k} + \Hnorm{\Psi}{k}
 \right)
 \nn
 \\
 &
 \quad
 +
 \blue{
 C
\left(
 \Hnorm{\p_T N}{k-2}
 +
 \Hnorm{\nhat}{k-2}
 +
 \Hnorm{X}{k-1} 
\right)
   \Hnorm{\jnull}{k-2} 
   }
 \nn
 \\
 &
 \quad
 +
 \blue{
 C
 \left(
  1
  +
  \Hnorm{\p_T N}{k-2}
 +
  \Hnorm{\p_T X}{k-2}
 +
 \Hnorm{\nhat}{k-2}
 +
  \Hnorm{X}{k-1} 
\right)
   |\tau| \cdot \Hnorm{X}{k-2} \cdot
   \Hnorm{\jvec}{k-2}
   }
 \nn
 \\
 &
 \quad
 +
 \blue{
 C |q|
 \Big\{
  |\tau| \cdot
   \Big[
      1
      +
     \Hnorm{\Gamma^\ast}{k-2}
     +
     \Hnorm{(\Sigma + g )X}{k-2}
      +
   \Hnorm{X}{k-2} \cdot  \Hnorm{\nhat}{k-2}
   }
 \nn
 \\
 &  
 \qquad\qquad
 +
 \blue{
   \Hnorm{\p_T X}{k-2}
   +
   \Hnorm{\p_T N}{k-2}
   +
   |q| \cdot  \Fnorm{F}{k-2}
   \Big]
   +
   \Hnorm{\Gamma^\ast_\ast}{k-2}
 \Big\}
 \energy{k-1}{\mu}
 }
 \,.
\end{align}
\begin{proof}
Differentiating \eqref{eq: elliptic of psi} in the direction of $\p_{e_0}$ and using the elliptic regularity, we arrive at (cf.~(5.39) in \cite{BFK19})
\begin{align}\label{eq: de0 of Psi, first equality}
 \Hnorm{\p_{e_0} \Psi}{k}
 &\leq
 C
 \Big(
  \Hnorm{ [\mcL_{e_0}, \Delta_g] \Psi}{k-2}
  +
  \Hnorm{\p_{e_0} \divg (\Psi D \log N)}{k-2}
 \nn
 \\
 &\qquad
  +
  \Hnorm{\p_{e_0}  [\mcL_{e_0}, \divg] \omega}{k-2}
  +
 \Hnorm{\p_{e_0} (N^{-1} \jnull + \tau \xhat^j \jvec_j) }{k-2} 
 \Big)
 \,.
\end{align}
The first three terms are estimated in  Lemma~5.9 in \cite{BFK19} by assuming that $\Hnorm{D \log N}{k}$ is small enough and $\Hnorm{\Pi}{k-2}$ is bounded:
\begin{align}
\lefteqn{  
  \Hnorm{ [\mcL_{e_0}, \Delta_g] \Psi}{k-2}
  +
  \Hnorm{\p_{e_0} \divg (\Psi D \log N)}{k-2}
  +
  \Hnorm{\p_{e_0}  [\mcL_{e_0}, \divg] \omega}{k-2}
}
&
\nn
\\
& \leq
 C
 \bigg(
  \Hnorm{ \Pi}{k-2}  
  +
  \Hnorm{ \mathcal{S} }{k-2}
  +
  \Hnorm{D \log N}{k-2}
  +
  \Hnorm{D \pnull \log N}{k-2}
 \nn
 \\
 &\qquad
 + 
  \Hnorm{\mcL_{e_0} \mathcal{S} }{k-2}
  +
  \Hnorm{\mcL_{e_0} \Pi }{k-2}  
 \bigg)
 \left(
  \sqrt{\EMax_k} + \Hnorm{\Psi}{k}
 \right)
 +
 \Hnorm{D \log N}{k-1}  \Hnorm{\pnull \Psi}{k-1}
 \,.\label{eq: first part of de0 Psi}
\end{align}
We estimate the last term of \eqref{eq: de0 of Psi, first equality}. To this end, we use the following relations:
\begin{align*}
 \p_{e_0} N^{-1}
 &=
 -N^{-3}
 \left(
  \p_T N
  +
  X^k \p_kN
 \right)
 \,,
 \\
 \p_{e_0} \tau
 &=
 -\tau N^{-1} 
 \,,
 \\
 \p_{e_0} X^k
 &=
N^{-1}
\left(
 \p_T X^k
 +
 X^j \p_j X^k
\right)
\,,
\\
\pnull ( \tau \xhat^j \jvec_j)
&=
- \tau X^k \jvec_k
+
\tau \pnull \xhat^k \jvec_k
+
\tau \xhat^k \pnull \jvec_k
\,,
\\
\Hnorm{ \p_{e_0} \left( N^{-1} \jnull \right)}{k-2}
&\leq
C \Hnorm{N^{-3}}{k-2}
\left(
 \Hnorm{\p_T N}{k-2}
 +
 \Hnorm{X}{k-2} \Hnorm{N}{k-1}
\right)
 \blue{
 \Hnorm{\jnull}{k-2} 
 }
\nn
\\
& \quad
+
C \Hnorm{N^{-1}}{k-2} \Hnorm{\pnull \jnull}{k-2}
\,,
\\
\Hnorm{\pnull \xhat}{k-2}
&\leq
C \Hnorm{N^{-3}}{k-2}
\big[
 \left(
  \Hnorm{\p_T N}{k-2}
  +
 \Hnorm{N}{k-1}   \Hnorm{X}{k-2}
\right) 
\Hnorm{X}{k-2}
\nn
\\
&
\qquad\quad\quad\quad\qquad
 +
 \Hnorm{N}{k-2}
  \left(
 \Hnorm{\p_T X}{k-2}
 +
 \Hnorm{X}{k-1} \Hnorm{X}{k-2}
\right)
\big]
\,,
\\
\Hnorm{\pnull ( \tau \scalpr{\xhat}{\jvec} )}{k-2}
&\leq
 C |\tau|
 \left(
  \Hnorm{N^{-1}}{k-2} \cdot  \Hnorm{\xhat}{k-2} 
  +
  \Hnorm{\pnull \xhat}{k-2}
\right)  
 \blue{
  \Hnorm{\jvec}{k-2}
 }
\nn
\\
& \quad
+  
C |\tau |
 \Hnorm{\xhat}{k-2} \Hnorm{\pnull \jvec}{k-2}
 \,.
\end{align*}
Then, by the smallness assumptions and boundedness of $\Hnorm{N}{k-2}$, one finds
\blue{
\begin{equation}\label{eq: second part of de0 Psi}
\begin{aligned}
\lefteqn{ 
 \Hnorm{\p_{e_0} (N^{-1} \jnull + \tau \xhat^j \jvec_j) }{k-2} 
}
&
%\nn
 \\ 
 &\leq
  C
\left(
 \Hnorm{\p_T N}{k-2}
 +
 \Hnorm{\nhat}{k-2}
 +
 \Hnorm{X}{k-1} 
\right)
   \Hnorm{\jnull}{k-2} 
% \nn
 \\
 & \quad
 +
 C
 \left(
  1
  +
  \Hnorm{\p_T N}{k-2}
 +
  \Hnorm{\p_T X}{k-2}
 +
 \Hnorm{\nhat}{k-2}
 +
  \Hnorm{X}{k-1} 
\right)
   |\tau| \cdot \Hnorm{X}{k-2} \cdot
   \Hnorm{\jvec}{k-2}
% \nn
 \\
 & \quad
 +
 C |q|
 \Big\{
  |\tau| \cdot
   \Big[
      1
      +
     \Hnorm{\Gamma^\ast}{k-2}
     +
     \Hnorm{(\Sigma + g )X}{k-2}
     +
   |q| \cdot  \Fnorm{F}{k-2}
      +
   \Hnorm{X}{k-2} \cdot  \Hnorm{\nhat}{k-2}
% \nn
 \\
 & \quad
   +
   \Hnorm{\p_T X}{k-2}
   +
   \Hnorm{\p_T N}{k-2}
   +
   |q| \cdot  \Fnorm{F}{k-2}
   \Big]
   +
   \Hnorm{\Gamma^\ast_\ast}{k-2}
 \Big\}
 \energy{k-1}{\mu}
 \,.
\end{aligned}
\end{equation}
}
Inserting \eqref{eq: first part of de0 Psi} and \eqref{eq: second part of de0 Psi} into \eqref{eq: de0 of Psi, first equality} and assuming that $\Hnorm{D \log N}{k-2}$ is small enough, one obtains the claim of the lemma.
\end{proof}
\end{lem}
\begin{prop}\label{prop: dTE of omega}
 For the energy defined in \eqref{eq: energy of omega} we have the estimate
\blue{
\begin{equation}
\begin{aligned}
 \p_T \EMax_k
 &\leq
  C
 \Big(
  \Hnorm{ \Sigma}{k-1}  
  +
  \Hnorm{\divg \Sigma}{k-1}
  +
  \Hnorm{ \pnull N}{k}
  +
  \Hnorm{N-3}{k}
%\nn
\\
 &  
 \qquad
  +
  \Hnorm{\mcL_{e_0} \Sigma}{k-2}
  +
  \Hnorm{\mcL_{e_0} \divg \Sigma}{k-2}  
 \Big)
 \EMax_k
% \nn
 \\
 & \quad
 +
 C
 \Big(
   \Hnorm{\p_T N}{k-2} 
   +
   \Hnorm{N-3}{k}
   +
   \Hnorm{\pnull N}{k}
   +
   \Hnorm{X}{k-1}
   +
   \Hnorm{\Sigma }{k-1}
%  \nn
 \\
 & \qquad \quad
  +
   \Hnorm{\divg \Sigma}{k-1}
   +
  \Hnorm{\mcL_{e_0} \Sigma}{k-2}
  +
  \Hnorm{\mcL_{e_0} \divg \Sigma}{k-2}  
 \Big) 
 \Hnorm{\jnull}{k-2} \sqrt{\EMax_k}
% \nn
 \\
 & \quad
 +
 C
 |\tau| \cdot \Hnorm{\jvec}{k-2} \sqrt{\EMax_k}
 +
 C |q|
 \Big(
  |\tau|
   +
   \Hnorm{\Gamma^\ast_\ast}{k-2}
 \Big)
 \energy{k-1}{\mu} \sqrt{\EMax_k}
 \,,
\end{aligned}
\end{equation}
}
as long as the norms in the brackets are bounded.
\begin{proof}
The boundedness of $\Hnorm{\Pi}{k-2}$ and \red{$\Hnorm{D \log N}{k-2}$ together with Lemma~\ref{lem-psi} } implies 
\blue{
\begin{equation}\label{eq: estimate of E plus Psi}
 \sqrt{\EMax_k} + \Hnorm{\Psi}{k}
 \leq
 C
 \left(
   \sqrt{\EMax_k}
   +
   \Hnorm{\jnull}{k-2}
   +
   |\tau| \cdot \Hnorm{X}{k-2} \Hnorm{\jvec}{k-2}
 \right)  
 \,.
\end{equation}
}
Then, combining  Lemmas~\ref{lem: dT of Ek} and \ref{lem: de0 of Psi}, applying the smallness assumptions, and finally assuming the boundedness of the norms appearing in the estimate yield the proof.
\end{proof}
\end{prop}
\begin{rem}
 Note that all terms appearing in the right-hand side of the inequality in Proposition~\ref{prop: dTE of omega} have good decay behaviour, i.e., they do not cause problems in the final estimate.
\end{rem}

\section{Energy estimates from divergence identity}\label{Sec: energy estimates divergence identity}

We define the $L^2$-Sobolev energy of the rescaled energy density by
\begin{equation}
 \bm{\varrho}_k 
 :=
 \sqrt{\sum_{\ell \leq k} \int_M
 \absolg{ D^\ell \rho }^2 \, d\mu_g}
 \,.
\end{equation}
We introduce 
\begin{equation}
 \widehat{\p}_T
 :=
 \p_T + \mcL_X
 \,.
\end{equation}
 Then, the divergence identity reads (cf.~\eqref{eq: divergence identity})
\begin{equation}
 \widehat{\p}_T \rho
 =
 (3 - N) \rho
 +
 \red{ \frac{1}{2} } \tau N^{-1} D_a (N^2 j^a)
 -
 \red{\frac{1}{6}} \tau^2 N g_{a b} T^{a b}
 -
 \red{\frac{1}{2}} \tau^2 N \Sigma_{a b} T^{a b}
 \,.
\end{equation}
Moreover, we need the following identities (cf.~\cite{CBMo01} and \cite{AF17})
\begin{equation}
 \big[\widehat{\p}_T, D_i \big] 
 D_{j_1} \ldots D_{j_m} u
 =
 - \sum_{ 1 \leq a \leq m}
 D_{j_1} \ldots D_{j_{a-1}} D_b \, D_{j_{a+1}} \ldots D_{j_m} u \cdot \widehat{\p}_T \Gamma^b_{{j_a} i }
 \,.
\end{equation}
and
\begin{equation}
 \p_T \int_M u \, d\mu_g
 =
 - \int_M (3 - N) u \, d\mu_g
 +
 \int_M \widehat{\p}_T u \, d\mu_g
\end{equation}
for any function $u$ on $M$. Then, one finds the following estimate for $\bmrho_k$.
\begin{prop}\label{Prop: bmrho}
For integers $k > 7/2$ the following estimate holds
\begin{equation}
\begin{aligned}
  \p_T \bm{\varrho}_k
 &\leq
 \red{
C \left(
    \Hnorm{D(N k)}{k - 2}
    +
    \Hnorm{N-3}{k}
  \right)
  }
  \bm{\varrho}_k
 \\
 &
 \quad
 +
 C \tau^2 \Hnormadj{N (\Sigma_{a b} + \tfrac{1}{3} g_{ab} ) T^{a b}}{k}
 +
 C |\tau| \cdot \Hnormadj{N^{-1} \divg (N^2 j)}{k}
 \,,
\end{aligned}
\end{equation}
\red{
provided that $|\tau| \mcG$ is bounded by a constant.
}
\begin{proof}
The proof is analogous to the proof of  Proposition~16 in \cite{AF17}.
\end{proof}
\end{prop}

\section{Elliptic Estimates for the lapse function and the shift vector}\label{Sec: elliptic estimates for X and N}

In this section we obtain elliptic estimates for the lapse function and the shift vector and their respective time derivatives.
\begin{prop}\label{prop: estimates for N and X}
The following estimates hold for the lapse function and the shift vector
\begin{align}
 \Hnorm{N-3}{k}
 &\leq 
 C
 \left(
  \Hnorm{\Sigma}{k-2}^2
  +
  |\tau| \cdot \Hnorm{\rho}{k-2}
  +
  |\tau|^3 \cdot \Hnorm{\underline{\eta}}{k-2}
 \right)
 \,,
 \\
 \Hnorm{X}{k}
 &\leq 
 C 
 \left(
  \Hnorm{\Sigma}{k-2}^2
  +
  \Hnorm{g-\gamma}{k-1}^2
  +
  |\tau| \cdot \Hnorm{\rho}{k-3}
  +
  |\tau|^3 \cdot \Hnorm{\underline{\eta}}{k-3}
  +
  \tau^2 \Hnorm{N j}{k-2}
 \right)
 \,.
\end{align}
\begin{proof}
 The estimates follows directly by the elliptic regularity applied to the elliptic equations for the lapse function and the shift vector.
\end{proof}
\end{prop}
\subsection{Estimate of the time derivatives}

Further, we can estimate the time derivatives of the lapse function and shift vector by using the elliptic estimate.
\begin{prop}\label{prop: estimates for dTN and dTX}
Let $k, \mu \in \N$ with $k > 9/2$ and $\mu \geq 3$. \red{For $\delta$-small data \eqref{eq: smallness assumption} with $\delta$ sufficiently small} the following estimates hold
\begin{align}
  \Hnorm{\p_T N}{k}
 &\leq 
 C
 \Big[
  \Hnorm{\nhat}{k}
  +
  \red{
  \Hnorm{X}{k}
  }
  +
  \Hnorm{\Sigma}{k-1}^2
  +
  \Hnorm{g - \gamma}{k}^2
  +
  |\tau| \cdot \Hnorm{S}{k-2}
   + 
   |\tau| \cdot \Hnorm{\rho}{k-1}
\nn
\\
 &
 \qquad
   +
   |\tau|^3 \cdot \Hnorm{ \underline{\eta} }{k-2}  
   +
   \tau^2 \Hnorm{ j }{k-1}
   +
   |\tau|^3 \cdot \Hnorm{ \underline{T} }{k-1}  
   + 
 |\tau|^3 \energy{k-1}{\mu + 1} 
 +
 \blue{
 \tau^2 \Fnorm{F}{k-1}^2
 }
 \Big]
  \,,
  \\
  \Hnorm{\p_T X}{k}
  &\leq
  C
  \red{
 \Big[
  \Hnorm{\nhat}{k-1}
  +
  \Hnorm{X}{k}
  +
  \Hnorm{\Sigma}{k-1}^2
  +
  \Hnorm{g - \gamma}{k-1}^2
  +
  |\tau| \cdot \Hnorm{S}{k-3}
   + 
   |\tau| \cdot \Hnorm{\rho}{k-2}
  }
\nn
\\
 & \qquad
 \red{
   +
   |\tau|^3 \cdot \Hnorm{ \underline{\eta} }{k-3}  
   +
   \tau^2 \Hnorm{ j }{k-1}
   +
   |\tau|^3 \cdot \Hnorm{ \underline{T} }{k-1}  
   + 
 |\tau|^3 \energy{k-2}{\mu + 1} 
  +
 \tau^2 \Fnorm{F}{k-2}^2
 }
 \Big]
 \,,
 \end{align}
where $\underline{T}$ denotes the spatial part of the rescaled energy-momentum tensor $T$ and integers $\mu \geq 3$ and $k > 7/2$.
\begin{proof}
Differentiating the elliptic system for $(N, X)$ and using the relations in Appendix \ref{App: time derivative of geom} yields
\begin{equation}
\begin{aligned}
 \left(
 \Delta - \tfrac{1}{3}
 \right) \p_T N
 &=
 2N \scalpr{D^2 N}{\Sigma} 
 \red{
 +
 }
 2 \nhat \Delta N
 \red{
 -
 }
 \scalpr{D^2 N}{\mcL_X g}
 \\
 &
 \quad
 +
 \red{
 \left[
  2 D_j (N \Sigma^{i j}) 
  -
  D^i \nhat
  -
  \Delta X^i
  +
  \tensor{R}{^i _j} X^j
 \right]
 }
 D_i N
 \\
 &
 \quad
 +
 \red{
 2 N
 \big[
 - 3 \left( N -\tfrac{1}{3} \right) \absolg{\Sigma}^2
 +
 2 \scalpr{\nabla X, \Sigma}{\Sigma}
 -
 N \scalpr{\Sigma}{\tfrac{1}{2} \mathcal{L}_{g, \gamma} (g - \gamma) + \mathbb{J} }
 }
 \\
 &
 \quad
 +
  \red{
 \scalpr{\Sigma}{D^2 N}
 -
 \scalpr{\Sigma}{\mcL_X \Sigma}
 +
 N \tau \scalpr{\Sigma}{S}
 \big]
 }
 \\
 &
 \quad
 +
 \p_T N
 \left(
  \absolg{\Sigma}^2 + \tau  \rho + \tau^3 \underline{\eta}
 \right)
 +
 N \left[\p_T(\tau  \rho) + \p_T (\tau^3 \underline{\eta}) \right]
 \,,
\end{aligned}
\end{equation}
where $\scalpr{\nabla X, \Sigma}{\Sigma} \equiv D_i X^j \Sigma^i_k \Sigma^k_j$, and
\begin{align}
 \Delta ( \p_T X^i )
 +
 \tensor{R}{^i_j} \, \p_T X^j 
 &=
 - \left[ \p_T, \Delta \right] X^i
 -
 \p_T \tensor{R}{^i_j}  X^j
 \nn
 \\
 &
 \quad
 +
 2 D_j (\p_T N) \Sigma^{i j}
 +
 2 D_j N \p_T \Sigma^{i j}
 -
 \p_T g^{i j} D_j \nhat
 -
 \tfrac{1}{3}  D^i (\p_T N)
 \nn
 \\
 &
 \quad
 +
 2 \p_T N \tau^2 j^i
 +
 2 N \p_T (\tau^2 j^i)
 -
 2
 \big(
   \p_T N \Sigma^{j k}
   +
   N \p_T \Sigma^{j k}
   -
   \p_T g^{j \ell} D_\ell X^k
 \nn
 \\
 &
 \quad
 -
 D^j  \p_T X^k
 -
 g^{j \ell} \p_T \Gamma^k _{m \ell} X^m
 \big) 
 \left(
   \Gamma^i_{j k} - \widehat{\Gamma}^i_{j k}
 \right)
 -
 2
 \left(
  N \Sigma^{j k} - D^j X^k
 \right)
 \p_T  \Gamma^i_{j k}
 \,. \label{eq: elliptic for dT X}
\end{align}
For the lapse function the elliptic regularity by using the smallness assumptions yields
\begin{equation}\label{eq: primary estimate of dT N}
\begin{aligned}
 \Hnorm{\p_T N}{k}
 &\leq 
 C
 \big[
  \Hnorm{\nhat}{k}
  +
  \Hnorm{X}{k}
  +
  \Hnorm{\Sigma}{k-1}^2
  +
  \Hnorm{g - \gamma}{k}^2
  +
  |\tau| \cdot \Hnorm{S}{k-2}
\\
 &
 \quad
   + 
   |\tau| \cdot \Hnorm{\rho}{k-2}
   +
   |\tau|^3 \cdot \Hnorm{ \underline{\eta} }{k-2}  
   +
   |\tau| \cdot \Hnorm{\p_T \rho}{k-2}
   +
   |\tau|^3 \cdot \Hnorm{ \p_T \underline{ \eta} }{k-2}  
 \\
 &
 \quad
 + 
  \left(
   \Hnorm{\Sigma}{k-2}^2
   +
   |\tau| \cdot \Hnorm{\rho}{k-2}
   +
   |\tau|^3 \cdot \Hnorm{ \underline{\eta} }{k-2}  
  \right)
  \Hnorm{\p_T N}{k-2}  
 \big]
 \,.
\end{aligned}
\end{equation}
Using elliptic regularity iteratively in conjunction with the smallness assumptions for both equations yields the estimates on the time-derivatives. The procedure is straighforward and follows the analogue case in \cite{AF17}.
\end{proof}
\end{prop}

\section{Energy Estimate For Geometric Objects}\label{Sec: energy estimate for geometric objects}

In this section we mainly adapt the results on the energy for the geometric perturbation from \cite{AM11} and \cite{AF17} to the present case.

\subsection{Decomposing the evolution equations}

We use the following form of the evolution equations.
\begin{lem}
The evolution equations for $g$ and $\Sigma$ are equivalent to the system 
\begin{subequations}\label{eq: decomposition of g and sigma}
\begin{align}
 \p_T (g - \gamma)
 &=
 2 N \Sigma
 +
 \mathcal{F}_{g - \gamma}
 \,,
 \\
 \p_T (6 \Sigma)
 &=
 -
 12 \Sigma
 -
 3 N \mathcal{L}_{g, \gamma} (g - \gamma)
 +
 6 N \tau  S
 -
 X^i \widehat{D}_i (6 \Sigma)
 +
 \mathcal{F}_\Sigma
 \,,
\end{align}  
\end{subequations}
where
\begin{align*}
 \Hnorm{\mathcal{F}_{g - \gamma}}{k}
 &\leq
 C
  \left(
     \Hnorm{\Sigma}{k-1}^2
     +
     \Hnorm{g-\gamma}{k}^2
     +
     |\tau| \cdot \Hnorm{\rho}{k-2}
     +
     |\tau|^3 \Hnorm{\underline{\eta}}{k-2}
     +
     \tau^2 \Hnorm{N j}{k-1}
  \right)
  \,,
  \\
 \Hnorm{\mathcal{F}_\Sigma}{k-1}
 &\leq
 C
  \left(
     \Hnorm{\Sigma}{k-1}^2
     +
     \Hnorm{g-\gamma}{k}^2
     +
     |\tau| \cdot \Hnorm{\rho}{k-1}
     +
     |\tau|^3 \Hnorm{\underline{\eta}}{k-1}
     +
     \tau^2 \Hnorm{N j}{k-2}
  \right)
  \,.
\end{align*}
\begin{proof}
The proof is formally identical to \cite{AF17}.
\end{proof}
\end{lem}

\subsection{Energy}
We define an energy for the trace-free part of the second fundamental form and the metric perturbation. This definition depends on the lowest eigenvalue $\lambda_0$ of the Einstein operator of $\gamma$. We define the constant $\alpha = \alpha(\lambda_0, \delta_\alpha)$ by
\begin{equation}
 \alpha := 1 - \delta_\alpha
 \,,
\end{equation}
where 
\begin{equation}\label{eq: delta alpha}
 \delta_\alpha
 :=
 \begin{cases}
  0
  \,, 
  & \lambda_0 > 1/9 \,,
  \\
  \sqrt{1 - 9 (\lambda_0 - \varepsilon)}
  \,,
  & \lambda_0 = 1/9
  \,,
 \end{cases}
\end{equation}
with $0 < \varepsilon \ll 1$. Next, we define the correction constant accordingly by
\begin{equation}
 c_E 
 :=
 1 - \delta_\alpha^2
 \,.
\end{equation}
Once $\varepsilon$ is fixed, in case $\lambda_0 = 1/9$, $\delta_\alpha$ can be made suitably small, independent of the other constants which play role in the final estimate.   

We are now able to define the energy for the geometric perturbation by
\begin{equation}
 E_k (g - \gamma, \Sigma)
 :=
 \sum_{1 \leq m \leq k} E_{ (m) } (g - \gamma, \Sigma)
 :=
 \sum_{1 \leq m \leq k} 
 \left[
  \mathcal{E}_{(m) } (g - \gamma, \Sigma) + c_E \Gamma_{(m)} (g - \gamma, \Sigma)
 \right]
 \,, 
\end{equation}
where 
\begin{equation}
\begin{aligned}
 \mathcal{E}_{(m) } (g - \gamma, \Sigma)
 &:=
 18 \int_M 
 \langle \Sigma, \mathcal{L}_{g, \gamma}^{m-1} \Sigma \rangle d\mu_g
 +
 \frac{9}{2} \int_M 
 \langle (g - \gamma), \mathcal{L}_{g, \gamma}^{m} (g - \gamma) \rangle d\mu_g
 \,,
 \\
 \Gamma_{(m)} (g - \gamma, \Sigma)
 &:=
 6 \int_M
 \langle \Sigma, \mathcal{L}_{g, \gamma}^{m - 1} (g - \gamma) \rangle d\mu_g
 \,,
\end{aligned}
\end{equation}
for integers $m \geq 1$. Note that here $\langle \cdot, \cdot \rangle$ is defined for the symmetric covariant $2$-tensors $u$ and $v$ by %(see Section~2 in \cite{AM2003} or Section~7 in \cite{AM11} for details)
$$
\langle u, v \rangle := u_{i j} v_{k \ell} \gamma^{i k} \gamma^{j \ell}
\,.
$$
\begin{lem}\label{lem: energy of geoemtry}
For the integers $k > 5/2$, there exists a  constant $C > 0$ such that \red{for $\delta$-small data \eqref{eq: smallness assumption} with $\delta$ sufficiently small} the inequality
\begin{equation}
  \Hnorm{g - \gamma}{k}^2
  +
  \Hnorm{\Sigma}{k-1}^2
  \leq
  C E_k (g - \gamma, \Sigma)
\end{equation}
holds.
\begin{proof}
For the proof we refer to  Lemma~19 in \cite{AF17}.
\end{proof}
\end{lem}
The following lemma gives an estimate for $E_k(g - \gamma, \Sigma)$.
Due to convenience we omit the argument $(g - \gamma, \Sigma)$ and write $E_k$ for the energy of the geometric perturbation in the following. 
\begin{lem}
\red{For $\delta$-small data \eqref{eq: smallness assumption} with $\delta$ sufficiently small}, the estimate
\begin{equation} \label{eq: T derivative of E geometry}
\begin{aligned}
 \p_T E_k
 &\leq
 - 2 \alpha E_k
 +
 C E_k^{3/2}
 \\
 &
 \quad
 +
 C E_k^{1/2}
  \left(
    |\tau| \cdot \Hnorm{NS}{k-1}
    +
    |\tau| \cdot \Hnorm{\rho}{k-1}
    +
    |\tau|^3 \Hnorm{\underline{\eta}}{k-1}
    +
  \tau^2  \red{ \Hnorm{Nj}{k-1} }
  \right)
\end{aligned}
\end{equation}
holds for integers $k > 5/2$.
\begin{proof}
The proof is the same as the proof of  Lemma~20 in \cite{AF17}.
\end{proof}
\end{lem}

\section{Total energy estimate}\label{Sec: total energy estimate}

We will, in the following, fine-tune the different energies to achieve sufficiently fast decay for the perturbation. This is done by introducing a total energy which incorporates the individual energies with some weight functions of time variable $T$. Then, under smallness assumptions on the total energy, the $L^2$-Sobolev energy of the rescaled energy density, and the support of the momentum we acquire an energy estimate for the total energy and in turn, for the individual energies.
\red{
Moreover, note that from now on we assume that the total charge $q$ is bounded by some constant.
}

\subsection{Total energy}

We first define the total energy in terms of the individual energies of the geometric perturbation, $L^2$-Sobolev energy of the distribution function, and the energy of the potential $1$-form with respective weight functions.
\begin{defn}[Total energy]
We define the total energy by
\begin{equation}
 \Etot(g-\gamma, \Sigma, f, \omega)
 :=
 e^{(1+\delta_E)T} E_6(g-\gamma, \Sigma)
 +
 e^{ - ( 1 - \delta_\mcE ) T} \mcE^2_{5,4}(f)
 +
 e^{- (1 - \delta_\EMax) T} \EMax_6 (\omega)
\,,
\end{equation}
where $\delta_E$, $\delta_\mcE$, and $\delta_\EMax$ are positive constants  such that $\delta_E+\delta_\mcE+\delta_\EMax \ll1$ and satisfy
\begin{equation}\label{eq: relations between constants}
\begin{aligned}
\delta_E & < \delta_\mcE 
\,,
\\
\delta_E &= \delta_\EMax - 2 \delta_\alpha
\,,
\end{aligned}
\end{equation}
which imply
\begin{equation}
2 \delta_\alpha  < \delta_\EMax < 2 \delta_\alpha + \delta_\mcE
\,.
\end{equation}
with $\delta_\alpha$ defined in \eqref{eq: delta alpha}.
\end{defn}

For the sake of simplicity we drop the argument $(g-\gamma, \Sigma, f, \omega)$ for the total energy in the following.

\subsection{Estimate for the Faraday tensor}

For the computations below we need an estimate for the Faraday tensor which is given in the following lemma.
\begin{lem}\label{lem: estimates on F}
\red{For $\delta$-small data \eqref{eq: smallness assumption} with $\delta$ sufficiently small} the following estimate holds
\begin{equation}
\Fnorm{F}{5}
\leq
C
e^{ (1 - \delta_E) T/2} \sqrt{ \Etot }
\,.
\end{equation}
\begin{proof}
The estimate follows from \eqref{eq: equivalence of F and E} and \eqref{eq: estimate of E plus Psi}, \red{and using Lemma~\ref{lem: J estimates} to estimate $\Hnorm{\jnull}{4}$ and $\Hnorm{\jvec}{4}$ by the $L^2$-energy of the distribution function}. Then, after estimating the result by the total energy we find 
\begin{equation}
\begin{aligned}
 \Fnorm{F}{5}
 &\leq
 C
 \left(
 \sqrt{ \EMax_6 }
  +
   \Hnorm{\jnull}{4}
   +
   |\tau|  \Hnorm{X}{4} \Hnorm{\jvec}{4}
 \right) 
\\
 &\leq
 C
 \left[
  e^{(1 - \delta_\EMax)T/2} 
  +
  e^{(1 - \delta_\mcE)T/2}
 \right]
 \sqrt{\Etot}
\\
 &\leq 
 C e^{(1 - \delta_E)T/2} 
 \left[
  e^{- \delta_\alpha T} 
  +
  e^{(\delta_E - \delta_\mcE)T/2}
 \right]
\sqrt{ \Etot }
 \,.
\end{aligned} 
\end{equation}
The terms in the last bracket are bounded using the conditions \eqref{eq: relations between constants}.
\end{proof}
\end{lem}

\subsection{Estimate for some geometric objects}

We summarize the estimates of the lapse function and the shift vector together with their time derivatives in the next lemma.
\begin{lem}\label{lem: Gamma and X by energy}
\red{For $\delta$-small data \eqref{eq: smallness assumption} with $\delta$ sufficiently small}, the estimates 
\begin{equation}
\begin{aligned}
\lefteqn{
\Hnorm{X}{k} + \Hnorm{N-3}{k} + \Hnorm{\p_T X}{k} + \Hnorm{\p_T N}{k} + \Hnorm{\Gamma^\ast_\ast}{k-1}
}
&
\\
&\leq
C
\left[
  e^{- (1 + \delta_E) T} \Etot
 +
 e^{ - (3 + \delta_\mcE ) T / 2 }  \sqrt{\Etot} 
 +
 e^{-T} \Hnorm{\rho}{k-1}
\right]
\end{aligned}
\end{equation}
and
\begin{equation}\label{eq: estimate of Gammas}
  \Hnorm{\Gamma^\ast}{k-1}
  \leq
  C
  \left[
  e^{- (1 + \delta_E) T}  \Etot 
   +
  e^{- (3 + \delta_\mcE ) T / 2} \sqrt{\Etot}
  +
  e^{-T} \bmrho_{k-1}
 \right]
\end{equation}
hold for integers $5 \leq k \leq 6$.
\begin{proof} 
We use Propositions~\ref{Prop: matter estimates}, \ref{prop: estimates for N and X}, and \ref{prop: estimates for dTN and dTX}, and  Lemma~\ref{lem: estimates on F} and apply them to \eqref{eq: Gammas} to get the estimates. Then, we estimate the result by the total energy. Now, we recall the constraints on the constants $(\delta_E, \delta_\mcE, \delta_\EMax)$ which lead to the result. 
Note that we do not estimate $\Hnorm{\rho}{k-1}$ in the first estimate. This term could be estimated later either by the $L^2$-energy of the distribution function or by the smallness assumption. On the other hand, in the second inequality we estimated $\Hnorm{\rho}{k-1}$ by $\bmrho_{k-1}$ which plays crucial role in the final estimate.
%Again note that we do not estimate $\Hnorm{\rho}{4}$ at this point.
\end{proof}
\end{lem}

\subsection{Estimate for $ \EMax_6$}

Combining the results above we find the following estimate for the time derivative of $\EMax_6$.
\begin{lem}\label{lem: T derivative of the EMax}
\red{For $\delta$-small data \eqref{eq: smallness assumption} with $\delta$ sufficiently small} the estimate holds
\begin{equation}\label{eq: T derivative of E Maxwell}
\p_T \EMax_6
\leq 
 C e^{-(\delta_E + \delta_\EMax) T} \Etot
 +
 C
 e^{-(1 + \delta_\EMax + 2 \delta_E ) T / 2} \Etot^{3/2}
 \,.
\end{equation}
\begin{proof}
The estimate follows directly from  Proposition~\ref{prop: dTE of omega} by invoking  Proposition~\ref{Prop: matter estimates},  Lemma~\ref{lem: Gamma and X by energy} and using the smallness assumptions, in particular, that of $\Hnorm{\rho}{4}$.
\end{proof}
\end{lem}

\subsection{Estimate for $\bmrho_4$}

For the rescaled energy density we prove the following lemma.
\begin{lem}\label{lem: bmrho estimate}
For $T_0 > 1$ and \red{$\delta$-small data \eqref{eq: smallness assumption} with $\delta$ sufficiently small} the following estimate holds
\begin{align}
 \bmrho_4  \Big|_T
 &\leq
 \left\{
  \bmrho_4 \Big|_{T_0}
 +
 C \int^T_{T_0}  
 \left[
  e^{ -  ( 1 + \delta_\mcE )  s / 2  }
  +
  e^{  - \delta_\EMax  s} \sqrt{  \Etot \Big|_s}
 \right]
 \sqrt{  \Etot \Big|_s} \, ds
\right\} 
\nn
\\ 
&
\qquad \qquad
\times
\exp \left\{
C \int^T_{T_0} 
 \left[
   e^{ - ( 1 +  \delta_E )  s / 2} 
   + 
   e^{ - ( 1 + \delta_\mcE )  s  /2 }
   +
   e^{ - \delta_\EMax \cdot s } \sqrt{  \Etot \Big|_s}
 \right]
 \sqrt{  \Etot \Big|_s} \, ds
 \right\}
 \,.
\end{align}
\begin{proof}
 From  Propositions~\ref{Prop: matter estimates} and ~\ref{Prop: bmrho}, we get under smallness assumptions
\begin{equation*}
 \p_T \bmrho_4 
 \leq
 C
 \left[
   \Hnorm{\Sigma}{3}
   +
   |\tau|  \cdot
    \EMax_3 
   + 
  |\tau| \cdot  \energy{2}{4}
 \right] \bmrho_4
 +
 |\tau|
 \left[
   \energy{5}{3}
   +
   \EMax_6
 \right]
 \,,
\end{equation*}
which after estimating by the total energy and integrating the result gives
\begin{align*}
 \bmrho_4 \Big|_T
 &\leq
 \bmrho_4 \Big|_{T_0}
 +
 C \int_{T_0}^T e^{-s} 
 \left[
  e^{   ( 1-\delta_\mcE )  s / 2  }
  +
  e^{ (1 - \delta_\EMax) s} \sqrt{  \Etot \Big|_s}
 \right]
 \sqrt{  \Etot \Big|_s} \,  ds
 \nn
 \\
 &
 \quad
 +
 C \int_{T_0}^T 
 \left[
   e^{ - ( 1 +  \delta_E )  s / 2} 
   + 
   e^{ - ( 1 + \delta_\mcE ) s /2 }
   +
   e^{ - \delta_\EMax s } \sqrt{  \Etot \Big|_s} 
 \right]
 \sqrt{  \Etot \Big|_s} \, \bmrho_4 \Big|_s \,  ds
 \,.
\end{align*}
Applying Gr\"onwall's inequality finishes the proof.
\end{proof}
\end{lem}

\subsection{Estimate for $\mcG$}

For the 	support of the momentum we get the following estimate.
\begin{lem}\label{lem: esitmate of G with Etot}
For $T_0 > 1$ and \red{$\delta$-small data \eqref{eq: smallness assumption} with $\delta$ sufficiently small} the following estimate holds
\begin{align}
 \mcG \Big|_T
 &\leq
 \left\{
 \mcG \Big|_{T_0}
 +
 C \int^T_{T_0} 
 \left[
  e^{- \delta_E s} \, \Etot\Big|_s
  +
\red{
  e^{ (1 - \delta_E) s/2} \sqrt{\Etot\Big|_s}
  }
  +
  \bmrho_4(s)
 \right] ds
 \right\}
\nn
\\
 &
 \qquad 
 \times
 \exp
 \left\{
   C \int^T_{T_0} 
 \left[
%   e^{- ( 1 + \delta_E)  s / 2} 
% +
 1
 +
   e^{- (1 + \delta_E) s/2} \sqrt{\Etot \Big|_s}
\right]  
e^{- (1 + \red{\delta_E} ) s / 2} \sqrt{\Etot \Big|_s} \,  ds
 \right\}
 \,.
\end{align}
\begin{proof}
 The proof follows directly from  Proposition~\ref{Prop: curly g} by using the estimates in Lemma~\ref{lem: Gamma and X by energy}.
\end{proof}
\end{lem}

\subsection{Estimate of the total energy}

We introduce a positive constant $\overline{C}$ which bounds all previous constants $C$ via
\begin{equation}
A C^3 \leq \overline{C}
\,,
\end{equation}
where $A > 1 $ is chosen suitably.

In this section we refer to the conditions \eqref{eq: assumption on tau G}--\eqref{eq: assumption on EMax6} below as bootstrap assumptions.
\begin{prop}\label{Prop: total energy}
\red{For $\delta$-small data \eqref{eq: smallness assumption} with $\delta$ sufficiently small} and the conditions
\begin{align}
 \overline{C} |\tau| \mcG 
 &\leq 
 \red{
 \epsilontot/6
 }
 \,, \label{eq: assumption on tau G}
 \\
 \overline{C} \bmrho_4 
 &\leq 
 \red{
 \epsilontot/6
 }
 \,, \label{eq: assumption on bmrho}
 \\
 \overline{C} \sqrt{\EMax_6}
 &\leq
  \red{
 \epsilontot/6
 }
 \,, \label{eq: assumption on EMax6}
\end{align}
and 
\begin{equation}\label{eq: assumption on the exp functions}
 \overline{C} e^{  (\delta_E - \delta_\mcE ) T }
% +
% \overline{C} e^{ - (3 - \red{1-\delta_\mcE} - \delta_\EMax ) T /2 }
 \leq
 \epsilontot / 4
 \,,
\end{equation}
for a positive small constant $\epsilontot$ and $T$ sufficiently large, there exists a  constant $0 < \bm{\epsilon} \ll 1$ such that 
\begin{equation}
 \p_T \Etot 
 \leq
 - 
 \left(
  1 - \bm{\epsilon}
 \right)
 \Etot  
 +
 \overline{C} \Etot^{3/2}  
 \,.
\end{equation}
\begin{proof}
We differentiate the total energy with respect to $T$, and use  Lemma~\ref{lem: energy of geoemtry},  Proposition~\ref{prop: T derivative of L2-energy of f}, and  Lemma~\ref{lem: T derivative of the EMax} for the energy of the perturbation of the geometry, the $L^2$-Sobolev energy of the distribution function, and the energy of the gauged vector potential, respectively. Then, we find
\begin{equation}\label{eq: T derivative of the total Energy 1}
\begin{aligned}
 \p_T\Etot
 &\leq
 - 
 (2\alpha - 1 - \delta_E) e^{(1 + \delta_E) T} E_6
 -
( 1-\delta_\mcE ) e^{ - ( 1-\delta_\mcE ) T} \mcE^2_{5, 4}
 -
 (1 - \delta_\EMax) e^{ - (1 - \delta_\EMax) T} \EMax_6 
 \\
 &
 \quad
 +
 \overline{C} 
 \left[
  e^{(1 + \delta_E) T} E_6
 \right]^{1/2}
  e^{  (\delta_E - \delta_\mcE ) T }
  \left[ 
  e^{ - ( 1-\delta_\mcE ) T / 2}
  \energy{5}{4}
  \right]
%  \nn
% \\
% && 
%  +
% \overline{C} e^{ (1 + \delta_E) T } E_6^{3/2} 
 \\
 &
 \quad
 +
 \overline{C} e^{- T } \Etot
 \\
 &
 \quad
 +
 \overline{C}
 \left(
  \bmrho_4
  +
  |\tau| \mcG
  +
  \red{
  \sqrt{\EMax_6}
  }
%  +
% \red{|\tau|} \cdot \Fnorm{F}{5}
 \right)
 e^{-  ( 1-\delta_\mcE )  T} \mcE^2_{5, 4}
 \\
 &
 \quad
 +
 \overline{C}   \Etot^{3/2}
 \,,
\end{aligned}
\end{equation}
where the first line comes from the $T$-derivatives of the exponentional functions which appear as factors of individual energies in the total energy and the first term of \eqref{eq: T derivative of E geometry}; the second line results from the second term in the bracket in the last line of \eqref{eq: T derivative of E geometry}; the third line results from the last term in \eqref{eq: T derivative of E Maxwell}; the fourth line comes from the terms in \eqref{eq: T derivative of L2-energy of f} which have less decay rates than the other ones.  In particular, the first term in the bracket comes from the estimate of $|\tau|^{-1} \Hnorm{N^{-1} \Gamma^\ast}{5}$ and the last term in the bracket comes from $\Fnorm{F}{5}$ using \eqref{eq: estimate of E plus Psi} and the estimates of  Lemma~\ref{lem: J estimates} which gives $\bmrho_4 + \sqrt{\EMax_6}$ and a higher order term which is absorbed in the fifth line.

Finally, the fifth line results from all other terms which are of higher-order in total energy, in particular the last term in the bracket in \eqref{eq: T derivative of L2-energy of f}.

The second, the third, and the fourth lines of \eqref{eq: T derivative of the total Energy 1} can be estimated using the conditions \eqref{eq: assumption on tau G}--\eqref{eq: assumption on the exp functions}. Hence, in total we have
\begin{align}
 \p_T\Etot
 &\leq
 - 
 (2\alpha - 1 - \delta_E) e^{(1 + \delta_E)} E_6
 -
  ( 1-\delta_\mcE )  e^{ -  ( 1-\delta_\mcE )  T} \mcE^2_{5, 4}
 -
 (1 - \delta_\EMax) e^{ - (1 - \delta_\EMax) T} \EMax_6 
 \nn
 \\
 &
 \quad
 +
 \frac{1}{2} \epsilontot \Etot
 +
 \frac{1}{2} \epsilontot  e^{-  ( 1-\delta_\mcE ) T} \mcE^2_{5, 4}
 +
 \overline{C}   \Etot^{3/2}
 \,.
\end{align}
Introducing a decay inducing positive constant $\sigma$, we rewrite the previous inequality as
\begin{equation}
\begin{aligned}
\p_T\Etot
 &\leq
 - (\sigma - \tfrac{1}{2} \epsilontot) \Etot
 + 
 ( \sigma - 2\alpha + 1 + \delta_E  ) e^{(1 + \delta_E)} E_6
 \\
 &
 \quad
 +
 ( \sigma - 1 + \delta_\mcE + \tfrac{1}{2} \epsilontot ) e^{ -  ( 1-\delta_\mcE ) T} \mcE^2_{5, 4}
 +
 ( \sigma - 1 + \delta_\EMax ) e^{ - (1 - \delta_\EMax) T} \EMax_6 
\\
&
\quad
 +
 \overline{C}   \Etot^{3/2}
 \,.
\end{aligned}
\end{equation}
Now, we choose $\sigma$ such that
$$
 \sigma = 1 -  \delta_\mcE  - \tfrac{1}{2} \epsilontot
 \,.
$$
Moreover, we require that
\begin{equation}\label{inequprev}
 4 \delta_\alpha \leq \epsilontot
 \,,
\end{equation}
which can always be achieved since $\delta_\alpha$ can be made sufficiently small such that \eqref{inequprev} holds, and in the case that $\lambda_0 = 1/9$ the inequality is trivial. These together with the condition \eqref{eq: relations between constants} imply
\begin{equation}
 \p_T\Etot
 \leq
 - ( 1 - \bm{\epsilon} ) \Etot 
 +
 \overline{C}   \Etot^{3/2}
 \,,
\end{equation}
with
\begin{equation}
 \bm{\epsilon} := \delta_\mcE + \epsilontot
 \,.
\end{equation}
\end{proof}
\end{prop}

\section{Global existence}\label{Sec: global existence}
We provide the complete global existence and stability argument based on the foregoing sections in the following.

\subsection{Launching a solution from small initial data}
We consider a not necessarily CMC initial data set for the EVMS close to initial data of the Milne spacetime. According the local existence theory in \cite{BCB73} there exists a local-in-time solution close to the Milne spacetime on a short time interval. As shown in \cite{FK15} then there exists a CMC surface in this spacetime, this is independent of the matter variables and hence applies similarly in the present setting. We then consider initial data $(g_0,\Pi_0,f_0,\omega_0,\dot\omega_0)$ on this CMC surface, which by continuity can be made arbitrarily small. We consider the unique solution launched from this initial data set in the following.

\subsection{Guaranteeing the smallness conditions \eqref{eq: assumption on tau G}--\eqref{eq: assumption on the exp functions} on an open interval }

Assuming that the solution exists up to $T_0$ and smallness at $T_0$ so that \eqref{eq: assumption on the exp functions} holds,
we choose new initial data at $T_0$ such that
\begin{equation}
 \Etot\Big|_{T_0}
 +
 \bmrho_4 \Big|_{T_0}
 +
 \mcG \Big|_{T_0}
 \leq
 \varepsilon_0
 \,,
\end{equation}
where $0 < \varepsilon_0 < C^{-1} \delta^2$ with $C$ being the maximum of all constants $C$ appeared so far.  Once $\varepsilon_0$ is chosen sufficiently small, one can decrease $\varepsilon_0$ further. The same is true for increasing $T_0$. The reason is that all estimates 	are uniform in the sense that they do not depend on the smallness of initial data.
By choosing $\varepsilon_0$ sufficiently small we make sure that conditions \eqref{eq: assumption on tau G} and \eqref{eq: assumption on bmrho} hold at $T_0$ and the condition \eqref{eq: assumption on the exp functions} holds at large $T$ automatically.

We define 
\begin{equation*}
\begin{aligned}
 T^\ast
 &:=
 \sup 
 \Big\{
  T > T_0 \Big| \text{the solution exists, is $\delta$-small,} 
  \\
 & \qquad\quad \qquad\qquad
 \text{and conditions   \eqref{eq: assumption on tau G} and \eqref{eq: assumption on bmrho} hold on $[T_0, T)$}
 \Big\}
 \,.
\end{aligned}
\end{equation*}

\subsection{Improving the bootstrap assumptions}
In the following we prove that if $\varepsilon_0$ is sufficiently small then $T^\ast = \infty$. This, on the one hand, proves the global existence and therefore the first part of Theorem~\ref{Theorem}, and on the other hand improves the bootstrap assumptions \eqref{eq: assumption on tau G} and \eqref{eq: assumption on bmrho} on $(T_0, T^\ast)$.

\begin{prop}
For any solution with initial data with sufficiently small $\varepsilon_0 > 0$ the conditions  \eqref{eq: assumption on bmrho} and \eqref{eq: assumption on tau G} hold on $(T_0, T^\ast )$. Moreover, we have $T^\ast = \infty$.

\begin{proof}
By the construction in the previous subsection, Proposition~\ref{Prop: total energy} holds on $(T_0, T^\ast)$ and implies
\begin{equation}
\frac{d \sqrt{\Etot }}{dT}
\leq 
 - \frac{1}{2} (1 - \bm{\epsilon}) \sqrt{\Etot}
 +
 \overline{C} \Etot
 \,.
\end{equation}
Assuming that
$$ 
\varepsilon_0 < \frac{1 - \bm{\epsilon}}{4 \overline{C}}
\,.  
$$
we find
\begin{equation}\label{eq: etot}
\begin{aligned}
 \sqrt{\Etot \Big|_T } 
 & \leq 
  \frac{1 - \bm{\epsilon} }{2 \overline{C} + \exp \left[ \tfrac{1}{2} (1 - \bm{\epsilon} ) (T - T_0) \right]
  \left[
    (1 - \bm{\epsilon} ) \Etot^{-1/2} \Big|_{T_0}
   - 
   2 \overline{C}
  \right] }
  \\
  &\leq
  (1 - \bm{\epsilon} )   \sqrt{\Etot \Big|_{T_0} } \,  e^{ - (1 - \bm{\epsilon} )(T - T_0) / 2 }
  \,.
\end{aligned}
\end{equation}

Combining this estimate with  Lemma~\ref{lem: bmrho estimate} yields
\begin{equation}
 \bmrho_4 \Big|_T
 \leq
 \left(
  \bmrho_4 \Big|_{T_0}
  +
  C \sqrt{\Etot \Big|_{T_0}} 
 \right)
 \exp 
  \left(
   C \sqrt{\Etot \Big|_{T_0}}
  \right)
 \leq  
 \left(
  \varepsilon_0 + C \sqrt{\varepsilon_0}
 \right)
  e^{C \varepsilon_0}
  \,.
\end{equation}
This implies \eqref{eq: assumption on bmrho} with a strict inequality by choosing $\varepsilon_0$ small enough.

Consequently, from the previous result, Lemmas~\ref{lem: Gamma and X by energy} and \ref{lem: esitmate of G with Etot},  and estimate of $\sqrt{\Etot\Big|_{T}}$ it follows
\begin{equation*}
\begin{aligned} 
 \mcG\Big|_T
 &\leq
 \left\{
  \mcG\Big|_{T_0}
  +
  C \int_{T_0}^T
  \left[
  e^{-  \delta_E s}  \Etot\Big|_s 
   +
  e^{- (1 + \delta_\mcE ) s / 2} \sqrt{\Etot\Big|_s}
  +
   \bmrho_4\Big|_s
 \right]
 ds
 \right\} 
 \\
 & \qquad \times
 \exp 
  \left\{
   C \int_{T_0}^T
   \left[
    e^{-  \delta_E s}  \Etot\Big|_s 
    +
    e^{- (1 + \delta_\mcE ) s / 2} \sqrt{\Etot\Big|_s} 
  \right]
  ds
  \right\}
  \\
 &\leq
 C \left[
  \varepsilon_0 
  +
   \sqrt{\varepsilon_0} (T - T_0)
  +
  \red{
  \sqrt{\varepsilon_0} \, e^{(\bm{\epsilon} - \delta_E)T/2}
  }
 \right]
 e^{C \sqrt{\varepsilon_0}}
 \,.
\end{aligned} 
\end{equation*}
Hence,
\begin{equation}
 \overline{C}|\tau| \mcG
 \leq
 \overline{C} C \sqrt{\varepsilon_0}
 \left[
    (T - T_0)
  +
  \red{
   e^{(\bm{\epsilon} - \delta_E)T/2}
  }
 \right]
  e^{-T}
 \,,
\end{equation}
which by choosing $\varepsilon_0$ small enough and for $T$ sufficiently large satisfies \eqref{eq: assumption on tau G} with a strict inequality.

Finally, we need to show that condition \eqref{eq: assumption on EMax6} holds with a strict inequality. Indeed, from  Lemma~\ref{lem: T derivative of the EMax} after inserting the estimate of $\Etot$ and integration on $(T_0, T^\ast)$ we have
\begin{equation}
 \EMax_6 \leq C \varepsilon_0
 \,,
\end{equation}
which again by choosing $\varepsilon_0$ sufficiently small satisfies the condition \eqref{eq: assumption on EMax6} with a strict inequality. 

Therefore,  conditions \eqref{eq: assumption on tau G}--\eqref{eq: assumption on EMax6} hold with strict inequality on $(T_0, T^\ast)$ by choosing $\varepsilon_0$ sufficiently small, and 
\begin{equation}
 \Etot\Big|_T 
 \leq
 (1 - \bm{\epsilon})^2 \, \Etot\Big|_{T_0} 
 \exp
  \left[
   - (1 - \bm{\epsilon}) \left( T - T_0 \right)
  \right]
  \,,
\end{equation}
on $(T_0, T^\ast)$. In particular, by the continuation criterion the solution can be extended to $ (T_0, T_0 + \epsilon) $ for a small positive constant $\epsilon$. The conditions \eqref{eq: assumption on tau G} and \eqref{eq: assumption on bmrho} hold on the extended interval by construction. Finally, a standard continuity argument shows that $T^\ast = \infty$. 
\end{proof}
\end{prop}

\subsection{Decay rates}

As a direct result of the decaying total energy we conclude the following decay rates
\begin{equation}\label{eq: decay rates}
\begin{aligned}
 \Hnorm{g - \gamma}{6}
 &\less
 \sqrt{\varepsilon_0}
 \exp
  \left[
   \left(
   - 1 + \frac{\bm{\epsilon} - \delta_E }{2}
   \right)
   T
  \right]
 \,,
 \\
 \Hnorm{\Sigma}{5}
 &\less
 \sqrt{\varepsilon_0}
 \exp
  \left[
   \left(
   - 1 + \frac{\bm{\epsilon} - \delta_E }{2}
   \right)
   T
  \right]
 \,,
 \\
 \Hnorm{N-3}{6}
 &\less
 \sqrt{\varepsilon_0}  \exp (- T)
 \,,
 \\
 \Hnorm{X}{6}
 &\less
 \sqrt{\varepsilon_0} \exp (- T)
 \,,
 \\
 \mcE_{5,4} (f)
 &\less
  \sqrt{\varepsilon_0} 
  \exp
  \left(
    \frac{\bm{\epsilon} - \delta_\mcE}{2} T
  \right)
  =
  \blue{
  \sqrt{\varepsilon_0} 
  \exp
  \left(
    \frac{\epsilontot}{2} T
  \right)
  }
 \,,
 \\
 \bmrho_4 
 &\less
\sqrt{\varepsilon_0}
 \,,
 \\
 \EMax_6 
 &\less \varepsilon_0
 \,.
\end{aligned}
\end{equation}
\red{Recall that $\varepsilon_0 \leq C^{-1} \delta^2$ which improves the bootstrap assumption on $\Fnorm{F}{5} < \delta$}. Thus, as a result of the obtained estimates the rescaled metric converges against the fixed Einstein metric $\gamma$, i.e., the geometry converges against the Milne universe towards future. The decay rates for the components of the spacetimes metric imply future completenes as in the foregoing works \cite{AF17,BFK19}.

\appendix

\section{Relations between two different connections and their associated geometric objects} \label{App: Diff Tensor}
\begin{defn}[Difference Tensor]\label{Def: Diff Tensor}
Let $\nabla$ and $\overline{\nabla}$ be the associated covariant derivatives of the metrics $g$ and $\bar{g}$ with Christoffel symbols $\Gamma$ and $\overline{\Gamma}$, respectively. Then, schematically, we have  
\begin{equation*}
 \nabla
 \equiv
 \overline{\nabla}
 +
 \Upsilon
 \,,
\end{equation*}
where the \emph{difference tensor} $\Upsilon$ is a $(1,2)$-tensor equal to $\Upsilon :=\Gamma -  \overline{\Gamma} $.
\end{defn}
\begin{rem}\label{rem: diff tensor component}
Applying $\nabla_i$ on the component of a vector field $v^j$ yields
$$
 \nabla_i v^j
 =
 \overline{\nabla}_i v^j
 +
 \difften{j}{k}{i} v^k
 \,,
$$
whereas for a component of a one-form $w_j$ it yields
$$
 \nabla_i w_j
 =
 \overline{\nabla}_i w_j
 -
 \difften{k}{j}{i} w_k
 \,,
$$
where in a local chart the components of the difference tensor defined in Definition~\eqref{Def: Diff Tensor} read
$$
 \difften{i}{j}{k}
 =
 \frac{1}{2}  g^{i \ell}
 \left(
   \overline{\nabla}_j g_{k \ell}
   +
  \overline{\nabla}_k g_{\ell j}
   -
   \overline{\nabla}_\ell g_{j k} 
 \right)
 \,.
$$
\end{rem}
\begin{rem}\label{rem: diff tensor conf}
 If there is a conformal relation between $g$ and $\bar{g}$, i.e., we have
 $$
  \bar{g} = e^{2 \Phi} g
  \,,
 $$
then, the difference tensor takes the following form
\begin{equation}
 \difften{\alpha}{\beta}{\gamma}
 =
 \delta^\alpha_\gamma \partial_\beta \Phi
 +
 \delta^\alpha_\beta \partial_\gamma \Phi
 -
 g_{\beta \gamma} g^{\alpha \lambda} \p_\lambda \Phi
 \,.
\end{equation}
Note that 
\begin{equation}
 \difften{\alpha}{\beta}{\alpha}
 =
 d\cdot
 \p_\beta \Phi
\,,
\end{equation}
where $d$ is the dimension of the manifold.
\end{rem}

\begin{rem}
Two Riemann curvature tensors of the two metrics $g$ and $\bar{g}$ are related to each other by the following equation
$$
 \tensor{\text{Riem}[g]}{^i_{j k \ell}}
 =
 \tensor{\text{Riem}[\bar{g}]}{^i_{j k \ell}}
 +
 \overline{\nabla}_k  \difften{i}{j}{\ell}
 -
 \overline{\nabla}_\ell  \difften{i}{j}{\ell}
 +
 \difften{i}{k}{m} \difften{m}{j}{\ell}
 -
 \difften{i}{\ell}{m} \difften{m}{j}{k}
 \,.
$$
Schematically, we have
$$
 \text{Riem}[g]
 \equiv
 \text{Riem}[\bar{g}]
 +
 \overline{\nabla}\Upsilon
 +
 \Upsilon^2
 \,.
$$
\end{rem}

\section{Time Derivative of some geometric objects}\label{App: time derivative of geom}

In this appendix we list evolution equations for some geometric objects  with respect to time variable T:
\begin{align}
\p_T g^{i j}
&=
 - 2 N \Sigma^{i j} 
 -
 2 \left( \tfrac{N}{3} - 1 \right) g^{i j}
 -
 \mcL_X g^{ i j}
 \,,
 \\
 \p_T \Gamma^i_{j k}
 &=
 D_j (N \Sigma^i_k)
 +
 D_k (N \Sigma^i_j)
 -
 D^i (N \Sigma_{j k} )
 \nn
 \\
 &
 \quad
 +
 D_j \nhat \delta^i_k
 +
 D_k \nhat \delta^i_j
 -
  D^i \nhat g_{j k}
 \red{
  -
  D_j D_k X^i  + \tensor{R}{^i_{ k j \ell}} X^\ell
  }
\,,
\\
\p_T \absolg{\Sigma}^2
&=
 2
 \big[
 - 3 \left( N -\tfrac{1}{3} \right) \absolg{\Sigma}^2
 +
 2 \scalpr{\nabla X, \Sigma}{\Sigma}
 -
 N \scalpr{\Sigma}{\tfrac{1}{2} \mathcal{L}_{g, \gamma} (g - \gamma) + \mathbb{J} }
 \nn
 \\
 &
 \quad
 +
 \scalpr{\Sigma}{D^2 N}
 -
 \scalpr{\Sigma}{\mcL_X \Sigma}
 +
 N \tau \scalpr{\Sigma}{S}
 \big]
 \,,
 \\
 \left[ \p_T, \Delta \right] X^i
 &=
 \p_T g^{a b} D_a D_b X^i
 +
 g^{a b}
 \left[
 D_a \left( \p_T \Gamma^i_{ j b} X^j \right)
 -
 \p_T \Gamma^j_{ a b} D_j X^i
 +
 \p_T \Gamma^i_{ j a} D_b X^j
 \right]
 \,.
\end{align}
And we summarize some results in terms of $ \widehat{\p}_T$ using the fact that \red{$\tensor{(\mcL_X \Gamma)}{^i_{ j k}} =   D_j D_k X^i  - \tensor{R}{^i_{ k j \ell}} X^\ell$};
\begin{subequations}
\begin{align}
  \widehat{\p}_T g^{i j}
  &=
  - 2 N \Sigma^{i j} 
  -
 2 \left( \tfrac{N}{3} - 1 \right) g^{i j}
 \,,
 \\
 \widehat{\p}_T  \Gamma^i_{j k}
 &=
 D_j (N k^i_k)
 +
 D_k (N k^i_j)
 -
 D^i (N k_{j k} )
 \,,
\end{align}
\end{subequations}

\section{Connection coefficients of the Sasaki metric}\label{App: connection coefficients of Sasaki}

The connection coefficients of the Sasaki metric $\mathbf{g}$ with respect to the frame $\{ \theta_{\mathbf{a}} \}_{\mathbf{a} \leq 6}$ read
\begin{equation}\label{eq: connection coefficient wrt h}
\begin{aligned}
 \mathbf{\Gamma}^a_{b c}
 &=
 \Gamma^a_{b c}
 \,,
  \\
  \mathbf{\Gamma}^a_{J b}
 &=
 \tfrac{1}{2} p^k \tensor{\text{Riem}[g]}{^a_{b k (J-3)}} 
 =
 \mathbf{\Gamma}^a_{ b J} 
 \,,
 \\
  \mathbf{\Gamma}^J_{b c}
 &=
  \tfrac{1}{2} p^k \tensor{\text{Riem}[g]}{^{J-3}_{k bc }} 
 \,,
\end{aligned}
\quad
\begin{aligned}
 \mathbf{\Gamma}^I_{J a}
 &=
  \Gamma^{I - 3}_{(J - 3) a}
 \,,
  \\
 \mathbf{\Gamma}^I_{a J}
 &=
 0
 =
  \mathbf{\Gamma}^a_{I J}
  \,,
  \\
   \mathbf{\Gamma}^I_{J K}
   &=
   0
   \,,
\end{aligned}
\end{equation}
where small letters stand for the numbers $\{1, 2, 3\}$ and capital letters take the values $\{4, 5, 6\}$.

\section{Commutation Relations}\label{App: Commutation}
The commutation relations between the horizontal and vertical derivatives on $TM$   are heavily used in the present paper. We gather them in the following:

\begin{equation}
\begin{aligned}
 [\A_a, \A_b ]
 &=
 p^\ell \, \tensor{\riem}{^k_{\ell b a}} \, \B_k
 \,,
 \\
 \left[ \A_a, \B_b \right] 
 &=
 \Gamma^k_{ba} \B_k
 \,,
 \\
 \left[ \B_a, \B_b \right]
 &=
 0
 \,,
 \\
 \left[ \B_j , p^i \B_i \right] f
 &=
 \B_j f
 \,,
 \\
 \left[ \A_j , p^i \B_i \right] f
 &=
 0
 \,,
 \\
 \left[ \mathbf{\Gamma}^\mathbf{a}_{\mathbf{b} \mathbf{c}}, p^i \B_i \right] f
 &=
 \begin{cases}
  - \mathbf{\Gamma}^\mathbf{a}_{\mathbf{b} \mathbf{c}}
  \quad
  &\text{if} \quad 
  \{ \mathbf{a} \geq 4 ; \mathbf{b}, \mathbf{c} \leq 3 \}
  \,\, \text{or} \,\,
   \{ \mathbf{a} \leq 3 ; 4 \leq \mathbf{b} + \mathbf{c} \leq 9 \}
   \,,
  \\
  0
   \quad 
  & \text{else}
  \,.
 \end{cases}
\end{aligned}
\end{equation}

\section{Derivatives of Momentum  functions}\label{App: momentum derivatives}

We summarize the horizontal and vertical derivatives on $TM$ applied to some momentum functions in the following:  
\begin{subequations}
\begin{align}
\A_a \absolg{p}
 &=
 0
 \,,
 \\ 
 \A_a
\scalpr{\xhat}{p}
 &=
 p^i D_a \xhat_i
 \,,
 \\
 \A_a \phat
 &=
 \frac{1}{\phat}
 \left[
  \tau^2 \scalpr{\xhat}{p} \, p^i D_a \xhat_i
  -
 \frac{1}{2}  
   \left(
    1 + \tau^2 \absolg{p}^2
   \right)  
   D_a |\xhat|^2
 \right]
 \,,
 \\
 \A_a \pbar
 &=
 -
 \frac{\A_a \phat - \tau \A_a \scalpr{\xhat}{p}}{\phat - \tau  \scalpr{\xhat}{p}}
 \cdot \pbar
 \,,
 \\
 \B_a p^0
 &=
 \frac{1}{\phat}
 \left(
  \tau \xhat_a p^0
  +
  \tau^2 N^{-1} \red{g_{a b} p^b} 
 \right)
 \,,
 \\
 \B_a \phat
 &=
 \frac{\tau^2}{\phat}
  \left[
   \scalpr{\xhat}{p} \xhat_a
   +
   \left( 1 - \absolg{\xhat}^2 \right)  \red{g_{a b} p^b} 
  \right]
 \\
 \B_a \mathcal{P}_k
 &=
 \frac{1}{\phat}
 \left(
  \tau \xhat_k \mathcal{P}_a - N^{-1} g_{k a} +   \mathcal{P}_k \,  \B_a \phat 
 \right)
 \,.\label{eq: B p0}
\end{align}
\end{subequations}
\begin{rem}
$
 \B_a p^0
 =
 -  \tau^2 \mathcal{P}_a
 \,.
$
\end{rem}
Moreover, one finds
\begin{align}
 \B_a 
  \Bigg(
    \frac{\absolgal{ p + \tau^{-1} p^0 X }^2}{\phat}
  \Bigg)
 &=
 \frac{2}{\phat} 
  \left(
    \red{g_{a b} p^b}  + \tau^{-1} p^0 X_a
  \right) 
  \left(
    1 + \frac{\tau}{N \phat} \scalpr{X}{p} + \frac{p^0}{N \red{\phat} } \absolg{X}^2
  \right)
  \nn
  \\
  &
  \quad
  -
  \tau^2
   \frac{\absolgal{ p + \tau^{-1} p^0 X }^2}{\phat^3}
  \left[
    \scalpr{\xhat}{p} \xhat_a + ( 1 - \absolg{X}^2 )  \red{g_{a b} p^b} 
  \right] 
  \,.
\end{align}

\section{Evolution equations for momentum functions}\label{app: T derivatives of momentum functions}
For the momentum functions defined in Subsection~\ref{subsub: mass-shell relations} we can compute their time derivatives
\begin{align}
 \p_T \phat
 &=
 \frac{1}{2 \phat}
 \Big[
  2 \tau^2 \scalpr{\xhat}{p}^2
  +
  2 \tau^2 \scalpr{\xhat}{p}
   \Big(
    \langle \xhat, p \rangle_{\dot{g}}
    +
    \tfrac{1}{N}
    \langle p, \p_T X - \xhat \p_T N \rangle_g
   \Big)
\nn
\\
&
\quad
 -
 \left( 1 + \tau^2 \absolg{p}^2 \right)
  \Big(
   \absol{\xhat}_{\dot{g}}^2
   +
   \tfrac{2}{N}
   \langle \xhat, \p_T X - \xhat \p_T N \rangle_g
  \Big)
  +
  \tau^2 ( 1 - \absolg{\xhat}^2 )
  \left(
   2 \absolg{p}^2 + \absol{p}^2_{\dot{g}}
  \right)
 \Big]
 \,,
\end{align}
and
\begin{align}
 \p_Tp^0
 &=
 2 p^0
 +
 \frac{1}{2 N \phat}
 \Big[
   4 (p^0)^2  (- N^2 + \absolg{X}^2)
   +
   6 p^0 \tau \scalpr{p}{X}
   +
   2\tau^2 \absolg{p}^2
 \nn
 \\
 &
 \quad
 +   
 (p^0)^2 \p_T(- N^2 + \absolg{X}^2)
 +
 2 \tau p^0 \langle p, X  \rangle_{\dot{g}} 
 +
 2 \tau p^0 \scalpr{p}{\p_T X} 
 +
 \tau^2 \absol{p}_{\dot{g}}^2
 \Big]
 \,.
\end{align}
Moreover, we have
\begin{equation}\label{eq: T derivative of mathcal P}
\begin{aligned}
 \p_T \mathcal{P}_k
 &=
 -\frac{1}{N \phat}
  \left[
    \tau^{-1}
    \left(
    X_k 
    +
    \p_T X_k 
    \right) 
    p^0
    +
    \tau^{-1} X_k \,  \p_T p^0
    +
    \left(
       \p_T g_{k \ell} 
       +
      2 g_{k \ell} 
    \right)   
     p^\ell
  \right]
 \\
 &
 \quad
 -
  \left(
   \frac{\p_T N}{N} + \frac{\p_T \phat}{\phat}
  \right) 
  \mathcal{P}_k
  \,,
\end{aligned}
\end{equation}
and
\begin{equation}
\begin{aligned}
 \p_T \absolgal{ p + \tau^{-1} p^0 X }^2
 &=
 \absolgdot{p + \tau^{-1} p^0 X }^2
 +
 4 \absolg{p}^2
 +
 4 \tau^{-1} \scalpr{p^0 X}{p}
 \\
 &
 \quad
 +
 2 \tau^{-1} \scalpr{p + \tau^{-1} p^0 X }{p^0 X}
 \\
 &
 \quad
 +
 2 \tau^{-1} \scalpr{ p + \tau^{-1} p^0 X }{ (\p_T p^0) X + p^0 \p_T X}
 \,.
\end{aligned}
\end{equation}
\section{Divergence Identity}
In terms of the unrescaled variables the continuity equations read (cf.~ \cite{Rendall-Book})
\begin{subequations}
\begin{align}
 \p_\tau \tilde{\rho}
 -
 \widetilde{X}^a D_a \tilde{\rho}
 -
 \widetilde{N} \tau \tilde{\rho}
 +
 \widetilde{N}^{-1} D_a (\widetilde{N}^2 \tilde{j}^a )
 -
 \widetilde{N} \tilde{k}_{a b} \mathbf{T}^{a b}
 &=
 0
 \,,
 \\ 
 \p_\tau \tilde{j}^a
 -
 \widetilde{X}^b D_b \tilde{j}^a
 -
 \tilde{j}^b \widetilde{D}^a \widetilde{X}_b 
 -
 \widetilde{N} \tau \tilde{j}^a
 +
 D_b ( \widetilde{N} \mathbf{T}^{a b} )
 -
 2 \widetilde{N} \tensor{\tilde{k}}{^a_b} \, \tilde{j}^b
 +
 \tilde{\rho} \, \widetilde{D}^a \tilde{N}
 &=
 0
 \,,
\end{align}  
\end{subequations}
where we used the notation $\widetilde{D}^a = \tilde{g}^{a b} D_b$.
With the definitions $\rho = 4 \pi \tilde{\rho} \cdot \tau^{-3}$, $j^a = 8 \pi \tilde{j}^a \cdot \tau^{-5}$, and $T^{a b } = \red{ 8 \pi } \mathbf{T}^{a b} \cdot \tau^{-7}$ the continuity equations take the following form
\begin{subequations}\label{eq: divergence identity}
\begin{align}
 \p_T \rho
 &=
 (3 - N) \rho
 -
 X^a D_a \rho
 +
 \red{\frac{1}{2} } \tau N^{-1} D_a (N^2 j^a)
 -
\red{ \frac{1}{6} } \tau^2 N g_{a b} T^{a b}
 -
\red{ \frac{1}{2} } \tau^2 N \Sigma_{a b} T^{ab}
 \,,\label{eq: continuity for rho}
 \\
 \p_T j^a
 &=
 \frac{5}{3}
 \left( 3 - N \right) j^a
 -
 X^b D_b j^a
 -
 j^b D^a X_b
 +
 \tau D_a (N T^{a b})
 -
 2 N \Sigma^a_b \, j^b
 +
 \red{2} \tau^{-1} \rho D^a N
 \,.\label{eq: continuity eq for j}
\end{align}  
\end{subequations}

\section{Time derivative of $\protect\underline{\eta}$ }

From the definition of $\underline{\eta}$ in \eqref{eq: etabar} we compute its time derivative by decomposing it into Vlasov and Maxwell parts $\p_T \underline{\eta} = \p_T\underline{\eta}^\text{V} + \p_T\underline{\eta}^\text{M}$, where
\begin{equation}\label{eq: T derivative of eta Vlasov}
\begin{aligned}
 \red{(4 \pi)^{-1}} \, \p_T \underline{\eta}^\text{V}
 &=
 \red{(4 \pi)^{-1}
 \left(
   3 \nhat   - D_i X^i + 6
 \right)
 \underline{\eta}^\text{V}
 }
 +
 \tau  \int
  \frac{p^a}{p^0} \A_a f \frac{\absolgal{ p + \tau^{-1} p^0 X }^2}{\phat}
  \, d\mu_p
 \\
 &
 \quad
 + 
 \p_T X^a 
 \Bigg\{
   2 \tau^{-1} 
   \int 
      f \frac{p^0}{\phat} 
      \left(
        \red{ g_{a b} p^b } + \tau^{-1} p^0 X_a 
      \right) 
    \, d\mu_p
 \\
 &
 \quad
 +
 4 \tau^{-1}
  \int f
     \frac{p^0}{2 N \phat^2} \scalpr{p + \tau^{-1} p^0 X}{X}
     \left(
       p^0 X_a + \tau \red{g_{a b} p^b}
     \right)
     d\mu_p
  \\
  &
  \quad
  -
  \int f 
  \frac{\absolgal{ p + \tau^{-1} p^0 X }^2}{N \phat^3}
  \left[
    \tau^2 \scalpr{\xhat}{p} \red{ g_{a b} p^b }
    -
    \left( 1 + \tau^2 \absolg{p}^2 \right) \xhat_a
  \right]
    d\mu_p
  \\
  &
  \quad
  \red{
  -
  \tau^{-1}  \int f  \,
  \B_a 
   \left(
      \frac{p^0}{ \phat}  \absolgal{ p + \tau^{-1} p^0 X }^2
   \right)
    d\mu_p
  }
 \Bigg\}\\
& \quad
 +
 \p_T N
 \Bigg\{
   \tau^{-1} \xhat^a \int f \,
    \B_a 
   \left(
      \frac{p^0}{ \phat}  \absolgal{ p + \tau^{-1} p^0 X }^2
   \right)
    d\mu_p
  \\
  &
  \quad
  -
  2 \tau^{-1} \int f
   \left(   \frac{p^0}{\phat} \right)^2
   \scalpr{p + \tau^{-1} p^0 X}{X} 
   \, d\mu_p
  \\
  &
  \quad
  +
  \int f
   \frac{\absolgal{ p + \tau^{-1} p^0 X }^2}{N \phat^3} 
   \left[
    \tau^2 \scalpr{\xhat}{p}^2 - \left( 1 + \tau^2 \absolg{p}^2  \right) \absolg{\xhat}^2
   \right]
    d\mu_p
 \Bigg\}
 \\
 &
 \quad
 +
 2 \tau^{-1} \int f
 \frac{\scalpr{p + \tau^{-1} p^0 X}{X}}{N \phat^2}
  \Big[
      2 (p^0)^2 \left( -N^2 + \absolg{X}^2  \right)
      +
      \tau p^0 \scalpr{p}{X}
      \\
      &
      \qquad\qquad
      -
      \tau^2 \absolg{p}^2
      +
       \frac{1}{2}  |X|^2_{\dot{g}}
      +
      \tau p^0 \langle p, X \rangle_{\dot{g}}
      +
      \frac{1}{2} \tau^2 |p|_{\dot{g}}^2
    \Big]
    d\mu_p
 \\
 &
 \quad
 +
 \int 
  \frac{f}{\phat}
   \left(
     \left|p + \tau^{-1} p^0 \red{X} \right|_{\dot{g}}^2
     +
     6 \tau^{-1} p^0 \scalpr{p + \tau^{-1} p^0 X}{X}
   \right)
    d\mu_p
 \\
 &
 \quad
 +
 2 \Gamma^i_j \int f p^j
  \B_i
   \left(
      \frac{ \absolgal{ p + \tau^{-1} p^0 X }^2 }{ \phat}
   \right)
    d\mu_p
 \\
 &
 \quad
  +
  \tau
  \left(
    \Sigma_{j k} + \tfrac{1}{3} g_{j k}
  \right)
  X^i \int f
   \B_i
   \left(
      \frac{p^j p^k}{ \pbar \phat}  \absolgal{ p + \tau^{-1} p^0 X }^2
   \right)
   d\mu_p
  \\
  &
  \quad
  \blue{
  +
  \tau q
 \int f \, 
  \left[
    \left(
    g^{i j} - \xhat^i \xhat^j
    \right)
    F_{a j}
    +
    \frac{\tau}{N} F_{a 0} \xhat^i
  \right]
  \B_i 
  \left(   
    \frac{p^a}{p^0}
  \right)  
  \frac{\absolgal{ p + \tau^{-1} p^0 X }^2}{\phat}
    \, d\mu_p
  }   
 \\
 &
 \quad
 \blue{
 +
  \tau q \int f \F^i
 \B_i
   \left(
    \frac{\absolgal{ p + \tau^{-1} p^0 X }^2}{\phat}
   \right)
    d\mu_p 
  } 
 \\
 &
 \quad
 +
 \int f
  \frac{\absolgal{ p + \tau^{-1} p^0 X }^2}{2 \phat^3}
  \Big[
    2 \tau^2 \scalpr{\xhat}{p}^2
    -
    2 \tau^2 \scalpr{\xhat}{p} \langle \xhat, p \rangle_{\dot{g}}
  \\
  &
  \qquad\qquad  
    +
    \left( 1 + \tau^2 \absolg{p}^2  \right) \absolgdot{X}^2
    +
    \tau^2 \left( 1 - \absolg{\xhat}^2  \right)
   \left( 2 \absolg{p}^2 - \absolgdot{p}^2 \right)
  \Big]
    d\mu_p   
 \\
 &
 \quad
 \red{
 +
 \tau^{-1} \left( \Gamma^a + \p_T X^a - \xhat^a \p_T N \right) \int f
  \B_a  
  \left(
      \frac{p^0}{ \phat}  \absolgal{ p + \tau^{-1} p^0 X }^2
   \right)
    d\mu_p  
   }
   \,,
\end{aligned}
\end{equation}
and
\begin{align}
 \blue{(4 \pi)^{-1}} \p_T \underline{\eta}^{\text{M}}
 &=
 \p_T g^{i j}
  \left[
    - \frac{1}{2} \tau N^{-2} F_{0 i} F_{0 j}
    -
    \frac{3}{2} N^{-2} X^k F_{0 i} F_{j k}
    +
    \frac{1}{4} \tau^{-1}
     \left(
      3 g^{k \ell} - 2 N^{-2} X^k X^\ell
     \right)
     F_{i k} F_{j \ell}
  \right]
 \nn
 \\
 &
 \quad
 +
 g^{i j}
  \Big[
    \tau N^{-3} 
     \left( \tfrac{1}{2} N + \p_T N \right) F_{0 i} F_{0 j}
    -
    \tau N^{-2} F_{0 i} \p_T F_{0 j}
    \nn
  \\
  &
  \quad
    +
    N^{-3} 
    \left( 3 \p_T N X^k - \tfrac{3}{2} N \p_T X^k \right) F_{0 i} F_{j k}
  -
  \frac{3}{2}  N^{-2} X^k
   \left( 
     \p_T F_{0 i} F_{j k} + F_{0 i} \p_T F_{j k}
   \right)
  \nn
  \\
  &
  \quad
  +
  \frac{1}{4} \tau^{-1}
  \left(
   3 g^{k \ell} 
   -
   2 N^{-2} X^k X^\ell
   +
   3 \p_T g^{k \ell}
   +
   4 N^{-3} \p_N X^k X^\ell
   -
   4 N^{-2} X^k \p_T X^\ell
  \right)
  F_{i k} F_{j \ell}
  \nn
  \\
  &
  \quad
  +
  \frac{1}{2} \tau^{-1}
  \left(
   3 g^{k \ell} - 2N^{-2} X^k X^\ell
  \right)
  F_{i k} \p_T F_{j \ell}
  \Big] 
  \,,\label{eq: T derivative of eta Maxwell}
\end{align}
where
\begin{align*}
 \p_T F_{0 i}
 &=
 - \tau^{-1} 
  \left(
    \p_T^2 \omega_i  - \p_i \p_T A_T 
  \right)
  \,,
\\
\p_T F_{i j}
 &=
 \p_i \p_T \omega_j - \p_j \p_T \omega_i
 \,,
\end{align*}
with
\begin{align*}
 \p_T A_T
 &=
  \red{N}^2 \p_{e_0} \Psi
  -
  N X^i \p_i \Psi
  +
  \p_T N \cdot  \Psi
  -
  \p_T X^i  \omega_i
  -
  X^i \p_T \omega_i
  \,,
  \\
  \p_T \omega_i
  &=
  N \mcL_{e_0} \omega_i
  -
  \mcL_X \omega_i
  \,,
  \\
  \p_T^2 \omega_i
  &=
  \red{ N^2 ( \mcL_{e_0} ( \mcL_{e_0} \omega))_i
  +
  \p_T N
  ( \mcL_{e_0}  \omega)_\ell 
  }
 \nn
 \\
 & \quad
  \red{
  -
  N 
 \left[
  X^i \p_i (\mcL_{e_0} \omega)_\ell
  +
  (\mcL_{e_0} \omega)_i  \p_\ell X^i
 \right]
 -
 \p_T
  \left(
    \omega_i \p_\ell X^i
    +
    X^i \p_i \omega_\ell
  \right)
  }
 \,.
\end{align*}
For the latter equations we used the relations between $A$, $\Psi$, $\omega_i$, and $e_0$.

\bibliographystyle{abbrv}

\bibliography{%
Milne-EVM-Final.bib%
}

\end{document}